\documentclass[12pt]{article}
\usepackage{amssymb}
\usepackage{amsmath}
\usepackage{amsmath}
\usepackage{amsthm}
\usepackage{amsfonts}
\usepackage{graphicx}
\usepackage{enumerate}
\usepackage{bbm}
\usepackage{dsfont}
\usepackage{natbib}
\usepackage{multirow}
\usepackage{curves}

\numberwithin{equation}{section}
\newtheorem{theorem}{Theorem}

\newtheorem{proposition}{Proposition}

\newtheorem{lemma}{Lemma}
\theoremstyle{definition}
\newtheorem{definition}{Definition}
\newtheorem{assumption}{Assumption}

\DeclareMathOperator{\E}{\text{E}}

\def\ind{\perp\!\!\!\perp}

\begin{document}

\title{Constructive Identification of Heterogeneous Elasticities in the
Cobb-Douglas Production Function}
\author{%
\begin{tabular}{ccc}
Tong Li\thanks{Tong Li. Department of Economics, Vanderbilt University, VU Station B \#351819, 2301 Vanderbilt Place, Nashville, TN 37235-1819. Email: tong.li@vanderbilt.edu. Phone: (615) 322-3582} & ${}$\qquad${}$ & Yuya Sasaki\thanks{Yuya Sasaki. Department of Economics, Vanderbilt University, VU Station B \#351819, 2301 Vanderbilt Place, Nashville, TN 37235-1819. Email: yuya.sasaki@vanderbilt.edu. Phone: (615) 343-3016.} 
\end{tabular}
}
\maketitle

\begin{abstract}\setlength{\baselineskip}{6mm}
This paper presents the identification of heterogeneous elasticities in the Cobb-Douglas production function. 
The identification is constructive with closed-form formulas for the elasticity with respect to each input for each firm. 
We propose that the flexible input cost ratio plays the role of a control function under ``non-collinear heterogeneity'' between elasticities with respect to two flexible inputs. 
The \textit{ex ante} flexible input cost share can be used to identify the elasticities with respect to flexible inputs for each firm. 
The elasticities with respect to labor and capital can be subsequently identified for each firm under the timing assumption admitting the functional independence. 
\end{abstract}


\section{Introduction}

\label{sec:introduction}

Heterogeneous output elasticities and non-neutral productivity are natural
features of production technologies. In addition, the heterogeneity and
non-neutrality are related to a number of empirical questions in development
economics, economic growth, industrial organization, and international
trade. The econometrics literature is relatively sparse about methods of
identifying production functions allowing for these empirically relevant
technological features. A couple of innovations have been made relatively
recently. By extending the approach of \cite{DoJa13}, \cite{DoJa15} propose
an empirical strategy to analyze constant elasticity of substitution (CES)
production function with labor augmenting productivity, which allows for
multi-dimensional heterogeneity and non-neutral productivity. By extending
the method of \cite{GaNaRi17}, \cite{KaScSu15} propose an empirical strategy
to analyze Cobb-Douglas production function with finitely supported
heterogeneous output elasticities.

We study identification of the Cobb-Douglas production function model with
infinitely supported heterogeneous coefficients indexed by unobserved latent
technologies. Our objective is to identify the vector of both non-additive
and additive parts of the productivity for each firm, where the non-additive
part consists of the output elasticity with respect to each input and the
additive part is the traditional neutral productivity. We provide
constructive identification with closed-form identifying formulas for each
of the heterogeneous output elasticities and the additive productivity for
each firm. Our constructive identification with closed-form formulas
provides a transparent argument in relation to potential identification
failures due to subtle yet critical issues, such as the functional
dependence problem pointed out by \cite{AcCaFr15} and the instrument
irrelevance problem pointed by \cite{GaNaRi17}.


\section{Relation to the Literature}

\label{sec:literature}

One of the challenges in empirical analysis of production functions is to
overcome the simultaneity in the choice of input quantities by rational
firms, which biases na\"{\i}ve estimates \citep{MaAn44}. The literature on
identification of production functions has a long history \citep[see
e.g.,][]{GrMa98,AcBeBePa07}, and remarkable progresses have been made over
the past two decades. Various ideas proposed in this literature facilitate
the identification result that we develop in this paper. Furthermore, this
literature has discovered some subtle yet critical sources of potential
identification failure, which we need to carefully take into account when we
construct our identification results. As such, it is useful to discuss in
detail the relations between our identification strategy and the principal
ideas developed by this literature.

A family of approaches widely used in practice today to identify parameters
of the Cobb-Douglas production function are based on control functions. \cite%
{OlPa96} propose to use the inverse of the reduced-form investment choice
function as a control function for latent technology. \cite{LePe03} proposes
to use the inverse of the reduced-form flexible input choice function as a
control function for the latent technology. See also \cite{Wo09} for
estimation of the relevant models. The main advantage of these
identification strategies is that an econometrician can be agnostic on the
form of the control function other than the requirement for the
invertibility of the function. Like the control function literature, we
employ a control function for latent technologies. However, unlike \cite%
{OlPa96} or \cite{LePe03}, we do not directly assume an invertible mapping
between an observed choice by firm and unobserved technology. Instead, we
only assume for construction of a control function that the ratio of
heterogeneous elasticities with respect to two flexible inputs are not
globally collinear -- see Assumption \ref{a:technology} ahead and
discussions thereafter. In other words, our approach requires
\textquotedblleft non-collinear heterogeneity\textquotedblright\ between
elasticities withe respect to two flexible inputs, in place of the
traditional assumption of invertible mapping.

\cite{AcCaFr15} point out the so-called functional dependence problem in the
control function approaches of \cite{OlPa96} and \cite{LePe03}, and propose
a few alternative structural assumptions to circumvent this problem. The
functional dependence problem refers to the rank deficiency for identifying
labor elasticity that arises because labor input that depends on the current
state variables loses data variations once the the state variables are fixed
through the control function. Among alternative structural assumptions to
avoid this problem, \cite{AcCaFr15} suggest a timing assumption where labor
input is determined slightly before the current state realizes -- also see \cite{AcHa15}. This
structural assumption is empirically supported by \cite{HuHuSa17}. The
structural assumptions (Assumptions \ref{a:flexible_input}--\ref%
{a:independence}) invoked in the present paper are consistent with the
timing assumption suggested by \cite{AcCaFr15}, and relevant data generating
processes can allow for the functional independence by a similar argument to
those of \cite{AcCaFr15} and \cite{KaScSu15}.

\cite{GaNaRi17} point out another source of identification failure in the
approach of using the flexible input choice function as a control function.
Namely, the Markovian model of state evolution which is commonly assumed in
this literature certainly induces orthogonality restrictions, but it also
nullifies the instrumental power or instrumental relevance for
identification of flexible input elasticities. Noting the role of
instruments from viewpoint of simultaneous equations, \cite{DoJa13,DoJa15}
suggest to use lagged input prices as alternative instruments assumed to
satisfy both the instrument independence and instrument relevance, and thus
solve this problem. In fact, for the Cobb-Douglas production functions, it
is long known that the first-order conditions and the implied input cost or
revenue shares inform us of input elasticities \citep{So57}. This approach
has been more recently revisited by \cite{Va03}, \cite{DoJa13,DoJa15}, \cite%
{KaScSu15}, \cite{GaNaRi17}, and \cite{GrLiZh16} in and beyond the context
of the Cobb-Douglas functions. The present paper also takes a similar
approach. We argue that the ratio of flexible input costs in conjunction
with the aforementioned assumption of non-collinear heterogeneity constructs
a control variable for the latent technology. Furthermore, the
\textquotedblleft \textit{ex ante} input cost share\textquotedblright\
defined as the share of input cost relative to the conditional expectation
of output value from firm's point of view is effective for constructive
identification of the heterogeneous elasticities with respect to flexible
inputs. This \textit{ex ante} input cost share is also directly identifiable
from data by econometricians once we construct the control variable from the
flexible input costs. The idea of using input cost share to identify
flexible input elasticity was hinted in \cite{GaNaRi17}, and we further
devise a way to extend this idea to models with non-additive productivity.
The model considered by \cite{GaNaRi17} and the model considered in the
present paper are complementary, in that the former is nonparametric with
additive productivity while the latter is linear with non-additive
productivity.

Productivity is sometimes treated as incidental parameters in panel data
analysis, but the literature on production functions has often circumvented
the incidental parameters problem \citep[cf.][]{NeSc48} via inversions of
maps representing choice rules of rational firms %
\citep[e.g.,][]{OlPa96,LePe03}. As already mentioned, we also circumvent
this problem via a control function based on the assumption of non-collinear
heterogeneity. Nonetheless, the existing methods to identify production
functions still utilize panel data to form orthogonality restrictions based
on the first difference in productivity %
\citep[e.g.,][]{OlPa96,LePe03,Wo09,AcCaFr15}. On the other hand, we do not
form such orthogonality restrictions based on panel data, as we can
explicitly identify the output elasticities with respect to flexible inputs
via the aforementioned \textit{ex ante} input cost shares. This aspect of
our approach is similar to that of \cite{GaNaRi17}.

While the literature on identification of production functions often
considers the CES productions functions (including the Cobb-Douglas and
translog approximation cases) with additive latent technologies, a departure from
Hicks-neutral productivity allows for answering many important economic
questions as emphasized in the introduction. \cite{DoJa15} extend the
identification strategy of \cite{DoJa13} to the framework of CES production
function with labor-augmenting technologies. The present paper shares
similar motivations to that of the preceding work by \cite{DoJa15}, but in
different and complementary directions. The labor-augmented CES production
function in the Cobb-Douglas limit case entails neutral productivity, and
hence the present paper focusing on non-neutral productivity in the
Cobb-Douglas production function attempts to complement the CES framework of 
\cite{DoJa15}. Cobb-Douglas production functions with non-additive
heterogeneity are studied in \cite{KaScSu15} and \cite{BaBrSa16}. \cite%
{KaScSu15} treat heterogeneous productivity via a mixture of the models of 
\cite{GaNaRi17}, and propose to identify the mixture components. On the
other hand, our framework allows for infinitely supported coefficients and
our method constructs identifying formulas for each coefficient for each
firm. For an application to international trade, \cite{BaBrSa16} consider
infinitely supported coefficients in the Cobb-Douglas production function,
where they almost directly assume identification for the moment restrictions
in a similar manner to \cite{AcCaFr15}, based on a multi-dimensional
invertibility assumption for the reduced-form flexible input choice. On the
other hand, the present paper develops the identification strategy instead
of assuming the identification, and complements \cite{BaBrSa16} by formally
establishing the identification result for a closely related model. We take
advantage of the first-order conditions, instead of relying on the
invertibility assumption, for the purpose of unambiguous identification of
output elasticities with respect to flexible inputs as emphasized earlier.


\section{The Model and Notations}

\label{sec:model}

Consider the gross-output production function in logarithm: 
\begin{equation}
y_{t}=\Psi \left( l_{t},k_{t},m_{t}^{1},m_{t}^{2},\omega _{t}\right) +\eta
_{t}\qquad \E[\eta_t]=0,  \label{eq:production_function}
\end{equation}%
where $y_{t}$ is the logarithm of output produced, $l_{t}$ is the logarithm
of labor input, $k_{t}$ is the logarithm of capital, $m_{t}^{1}$ is the
logarithm of a flexible input such as materials, $m_{t}^{2}$ is the
logarithm of another flexible input such electricity, $\omega _{t}$ is an
index of latent technology, and $\eta _{t}$ is an idiosyncratic productivity
shock. The Cobb-Douglas production function takes the form 
\begin{equation}
\Psi \left( l_{t},k_{t},m_{t}^{1},m_{t}^{2},\omega _{t}\right) =\beta
_{l}(\omega _{t})l_{t}+\beta _{k}(\omega _{t})k_{t}+\beta _{m^{1}}(\omega
_{t})m_{t}^{1}+\beta _{m^{2}}(\omega _{t})m_{t}^{2}+\beta _{0}(\omega _{t})
\label{eq:short_hand}
\end{equation}%
with heterogeneous coefficients $(\beta _{l}(\omega _{t}),\beta _{k}(\omega
_{t}),\beta _{m^{1}}(\omega _{t}),\beta _{m^{2}}(\omega _{t}),\beta
_{0}(\omega _{t}))$ that are indexed by the latent technology $\omega _{t}$.
The latent technology $\omega _{t}$ affects the additive productivity $\beta
_{0}(\omega _{t})$ and the elasticities $\beta _{l}(\omega _{t})$, $\beta
_{k}(\omega _{t})$, $\beta _{m^{1}}(\omega _{t})$, and $\beta
_{m^{2}}(\omega _{t})$ in non-parametric and non-linear ways.
Econometricians may not know the functional forms of $\beta _{l}(\cdot )$, $%
\beta _{k}(\cdot )$, $\beta _{m^{1}}(\cdot )$, $\beta _{m^{2}}(\cdot )$, or $%
\beta _{0}(\cdot )$.

Let $p_{t}^{y}$ denote the unit price of the output faced by firm $j$ at
time period $t$. Let $p_{t}^{m^{1}}$ and $p_{t}^{m^{2}}$ denote the unit
prices of the two types ($k=1$ and $2$) of the flexible input, respectively.
(We remark that these prices need not be observed in data for our
identification argument. We only require to observe the output value, $%
p_{t}^{y}\exp (y_{t})$, and input costs, $p_{t}^{m^{1}}\exp (m_{t}^{1})$ and 
$p_{t}^{m^{2}}\exp (m_{t}^{2})$, for our identification results.) With these
notations, we make the following assumption on flexible input choice by
firms.

\begin{assumption}[Flexible Input Choice]
\label{a:flexible_input} A firm at time $t$ with the state variables $%
(l_{t},k_{t},\omega_{t})$ chooses the flexible input vector $%
(m^1_{t},m^2_{t})$ by the optimization problem 
\begin{eqnarray*}
\max_{(m^1,m^2) \in \mathbb{R}_+^2} p_{t}^y \exp\left(\Psi\left(
l_{t},k_{t},m^1,m^2,\omega_{t} \right)\right) E[\exp\left(\eta_{t}\right)]
-p_{t}^{m^1} \exp\left( m^1 \right) -p_{t}^{m^2} \exp\left( m^2 \right),
\end{eqnarray*}
where $p_{t}^{m^1} > 0$ and $p_{t}^{m^2} > 0$ almost surely.
\end{assumption}

This assumption consists of two parts. First, each firm makes the choice of
the flexible input vector $(m_{t}^{1},m_{t}^{2})$ by the expected profit
maximization against unforeseen shocks $\eta _{t}$ given the state variables 
$(l_{t},k_{t},\omega _{t})$. Second, firms almost surely face strictly
positive flexible input prices. Furthermore, in order to guarantee the
existence of these flexible input solutions, we make the following
assumption of diminishing returns with respect to flexible input.

\begin{assumption}[Finite Solution]
\label{a:finite} $\beta_{m^1}(\omega_{t})+\beta_{m^2}(\omega_{t}) < 1$
almost surely.
\end{assumption}

Note that this assumption only requires diminishing returns with respect to
the subvector $(m^1_{t},m^2_{t})$ of only flexible inputs, and not
necessarily with respect to the entire vector $(l_{t},k_{t},m^1_{t},m^2_{t})$
of production factors. Assumptions \ref{a:measurable} and \ref{a:technology}
stated below concern about the heterogeneous elasticity functions $%
(\beta_{l}(\cdot),\beta_{k}(\cdot),\beta_{m^1}(\cdot),\beta_{m^2}(\cdot),%
\beta_{0}(\cdot))$, and involve requirements for a number of functions to be
measurable.

\begin{assumption}[Measurability]
\label{a:measurable} $\beta_{l}(\cdot)$, $\beta_{k}(\cdot)$, $%
\beta_{m^1}(\cdot)$, $\beta_{m^2}(\cdot)$ and $\beta_{0}(\cdot)$ are
measurable functions.
\end{assumption}

Assumption \ref{a:measurable} is satisfied if, for example, $\beta_{l}(\cdot)
$, $\beta_{k}(\cdot)$, $\beta_{m^1}(\cdot)$, $\beta_{m^2}(\cdot)$ and $%
\beta_{0}(\cdot)$ are continuous functions of the latent technology $%
\omega_{t}$. In particular, for the additive-productivity models considered
in the literature \citep[e.g.,][]{OlPa96,LePe03,Wo09}, this assumption is
trivially satisfied as $\beta_{l}(\cdot)$, $\beta_{k}(\cdot)$, $%
\beta_{m^1}(\cdot)$, and $\beta_{m^2}(\cdot)$, being constant functions
(i.e., $\beta_{l}(\omega_{t}) \equiv \beta_{l}$, $\beta_{k}(\omega_{t})
\equiv \beta_{k}$, $\beta_{m^1}(\omega_{t}) \equiv \beta_{m^1}$, and $%
\beta_{m^2}(\omega_{t}) \equiv \beta_{m^2}$) are continuous, and $%
\beta_{0}(\cdot)$ being the identity function (i.e., $\beta_{0}(\omega_{t})
\equiv \omega_{t}$) is continuous as well.

\begin{assumption}[Non-Collinear Heterogeneity]
\label{a:technology} The function $\omega \mapsto \beta_{m^1}(\omega) /
\beta_{m^2}(\omega)$ is measurable and invertible with measurable inverse.
\end{assumption}

\begin{figure}[tbp]
\centering
\setlength{\unitlength}{0.6cm} 
\begin{tabular}{cc}
\begin{picture}(12,12) \put(0,10){(a) \ \ Assumption \ref{a:technology} is
satisfied.} \multiput(1,1)(1,0){1}{\vector(0,1){8}}
\multiput(1,1)(1,0){1}{\vector(1,0){8}} \put(9,0.5){$\beta_{m^1}$}
\put(-0.2,9){$\beta_{m^2}$} \linethickness{0.25mm} \qbezier(1,2)(5,4)(9,6)
\put(6.8,4.5){\scriptsize$\omega_{t} \mapsto (\beta_{m^1}(\omega_{t}),$}
\put(8.3,4,0){\scriptsize$\beta_{m^2}(\omega_{t}))$} \linethickness{0.1mm}
\curvedashes[0.5mm]{1,1} \qbezier(1,1)(4,5)(7,9) \qbezier(1,1)(5,5)(8,8)
\qbezier(1,1)(5,4)(9,7) \put(1.8,2.7){\scriptsize$\omega'$}
\put(2.8,3.2){\scriptsize$\omega''$} \put(4.75,4.2){\scriptsize$\omega'''$}
\put(2.19,2.59){\circle*{0.25}} \put(3.00,3.00){\circle*{0.25}}
\put(5.03,4.03){\circle*{0.25}} \end{picture} & \begin{picture}(12,12)
\put(0,10){(a$^\prime$) \qquad Assumption \ref{a:technology} fails.}
\multiput(1,1)(1,0){1}{\vector(0,1){8}}
\multiput(1,1)(1,0){1}{\vector(1,0){8}} \put(9,0.5){$\beta_{m^1}$}
\put(-0.2,9){$\beta_{m^2}$} \linethickness{0.25mm} \qbezier(1,1)(5,4)(9,7)
\put(6.0,4.3){\scriptsize$\omega_{t} \mapsto (\beta_{m^1}(\omega_{t}),$}
\put(7.5,3,8){\scriptsize$\beta_{m^2}(\omega_{t}))$} \linethickness{0.1mm}
\curvedashes[0.5mm]{1,1} \qbezier(1,1)(4,5)(7,9) \qbezier(1,1)(5,5)(8,8)
\qbezier(1,1)(5,4)(9,7) \put(4.8,4.2){\scriptsize$\omega'$}
\put(5.9,5.1){\scriptsize$\omega''$} \put(7.2,6.1){\scriptsize$\omega'''$}
\put(5,4){\circle*{0.25}} \put(6.35,5){\circle*{0.25}}
\put(7.7,6){\circle*{0.25}} \end{picture} \\ 
\begin{picture}(12,12) \put(0,10){(b) \ \ Assumption \ref{a:technology} is
satisfied.} \multiput(1,1)(1,0){1}{\vector(0,1){8}}
\multiput(1,1)(1,0){1}{\vector(1,0){8}} \put(9,0.5){$\beta_{m^1}$}
\put(-0.2,9){$\beta_{m^2}$} \linethickness{0.25mm} \qbezier(1,3)(5,3)(9,6)
\put(6.5,4.0){\scriptsize$\omega_{t} \mapsto (\beta_{m^1}(\omega_{t}),$}
\put(8.0,3,5){\scriptsize$\beta_{m^2}(\omega_{t}))$} \linethickness{0.1mm}
\curvedashes[0.5mm]{1,1} \qbezier(1,1)(4,5)(7,9) \qbezier(1,1)(5,5)(8,8)
\qbezier(1,1)(5,4)(9,7) \put(2.36,3.31){\scriptsize$\omega'$}
\put(3.06,3.46){\scriptsize$\omega''$}
\put(4.15,3.75){\scriptsize$\omega'''$} \put(2.56,3.11){\circle*{0.25}}
\put(3.26,3.26){\circle*{0.25}} \put(4.45,3.55){\circle*{0.25}} \end{picture}
& \begin{picture}(12,12) \put(0,10){(b$^\prime$) \qquad Assumption
\ref{a:technology} fails.} \multiput(1,1)(1,0){1}{\vector(0,1){8}}
\multiput(1,1)(1,0){1}{\vector(1,0){8}} \put(9,0.5){$\beta_{m^1}$}
\put(-0.2,9){$\beta_{m^2}$} \linethickness{0.25mm} \qbezier(1,3)(9,3)(5,9)
\put(5.5,3.5){\scriptsize$\omega_{t} \mapsto (\beta_{m^1}(\omega_{t}),$}
\put(7.0,3,0){\scriptsize$\beta_{m^2}(\omega_{t}))$} \linethickness{0.1mm}
\curvedashes[0.5mm]{1,1} \qbezier(1,1)(4,5)(7,9) \qbezier(1,1)(5,5)(8,8)
\qbezier(1,1)(5,4)(9,7) \put(2.20,3.24){\scriptsize$\omega'$}
\put(2.95,3.35){\scriptsize$\omega''$}
\put(3.90,3.55){\scriptsize$\omega'''$} \put(2.50,3.05){\circle*{0.25}}
\put(3.15,3.15){\circle*{0.25}} \put(4.15,3.35){\circle*{0.25}}
\put(6.06,7.26){\scriptsize$\omega''''''$}
\put(6.48,6.08){\scriptsize$\omega'''''$}
\put(6.40,4.70){\scriptsize$\omega''''$} \put(5.86,7.46){\circle*{0.25}}
\put(6.28,6.28){\circle*{0.25}} \put(6.20,4.90){\circle*{0.25}} \end{picture}%
\end{tabular}
${}$%
\caption{Illustrations of Assumption \protect\ref{a:technology}. The bold
solid curves indicate the paths representing the technology $\protect\omega%
_{t} \mapsto (\protect\beta_{m^1}(\protect\omega_{t}),\protect\beta_{m^2}(%
\protect\omega_{t}))$. The dashed rays from the origin indicate linear
paths. The left column, (a) and (b), of the figure illustrates cases that
satisfy Assumption \protect\ref{a:technology}. In these graphs, the ratio $%
\protect\beta_{m^1}(\protect\omega) / \protect\beta_{m^2}(\protect\omega)$
is associated with a unique value of $\protect\omega$ provided an injective
technology $\protect\omega_{t} \mapsto (\protect\beta_{m^1}(\protect\omega%
_{t}),\protect\beta_{m^2}(\protect\omega_{t}))$. The right column, (a$%
^{\prime }$) and (b$^{\prime }$), of the figure illustrates cases that
violate Assumption \protect\ref{a:technology}. In these graphs, the ratio $%
\protect\beta_{m^1}(\protect\omega) / \protect\beta_{m^2}(\protect\omega)$
is not associated with a unique value of $\protect\omega$ even if the
technology $\protect\omega_{t} \mapsto (\protect\beta_{m^1}(\protect\omega%
_{t}),\protect\beta_{m^2}(\protect\omega_{t}))$ is injective.}
\label{fig:a:technology}
\end{figure}
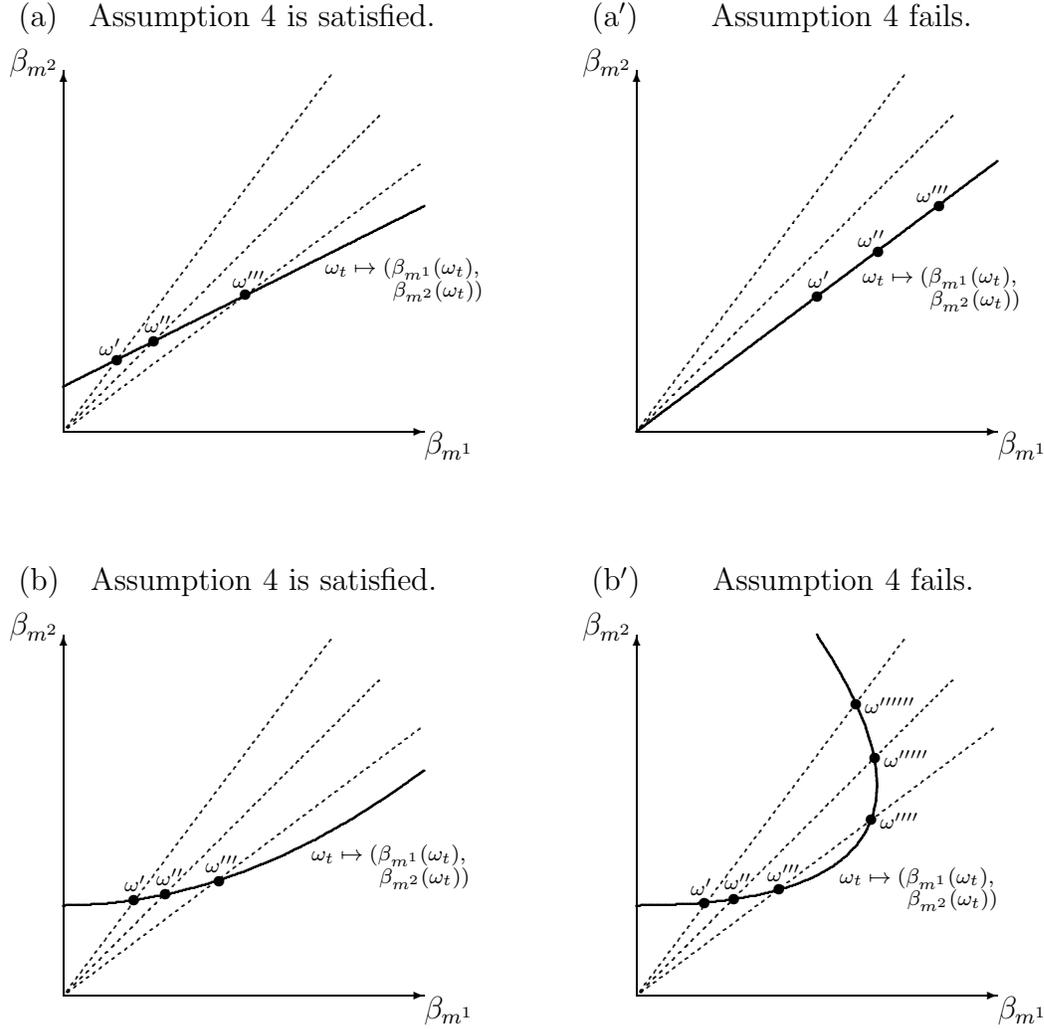

Assumption \ref{a:technology} is satisfied if the rate at which the latent
technology $\omega_{t}$ increases the output elasticity $\beta_{m^1}(%
\omega_{t})$ with respect to $m^1$ is strictly higher or strictly lower than
the rate at which the latent technology $\omega_{t}$ increases the output
elasticity $\beta_{m^2}(\omega_{t})$ with respect to $m^2$. Figure \ref%
{fig:a:technology} provides a geometric illustration of Assumption \ref%
{a:technology}. The bold solid curves indicate the technological paths $%
\omega_{t} \mapsto (\beta_{m^1}(\omega_{t}),\beta_{m^2}(\omega_{t}))$. The
dashed rays from the origin indicate linear paths. The left column, (a) and
(b), of the figure illustrates cases that satisfy Assumption \ref%
{a:technology}. In these graphs, the ratio $\beta_{m^1}(\omega) /
\beta_{m^2}(\omega)$ is associated with a unique value of $\omega$ provided
an injective technology $\omega_{t} \mapsto
(\beta_{m^1}(\omega_{t}),\beta_{m^2}(\omega_{t}))$. Assumption \ref%
{a:technology} is a nonparametric shape restriction, as opposed to a
parametric functional restriction, and hence both a simple linear case (a)
and a nonlinear case (b) satisfy this assumption. The right column, (a$%
^{\prime }$) and (b$^{\prime }$), of the figure illustrates cases that
violate Assumption \ref{a:technology}. In these graphs, the ratio $%
\beta_{m^1}(\omega) / \beta_{m^2}(\omega)$ is not associated with a unique
value of $\omega$ even if the technology $\omega_{t} \mapsto
(\beta_{m^1}(\omega_{t}),\beta_{m^2}(\omega_{t}))$ is injective.

Note the difference between Assumption \ref{a:technology} and the
invertibility assumptions used in the nonparametric control function
approaches \citep[e.g.,][]{OlPa96,LePe03}, where an invertible mapping
between the technology $\omega_{t}$ and an observable, such as investment
choice or flexible input choice, is assumed to exist. Assumption \ref%
{a:technology} does not assume such an inversion between the technology and
observed choices by firms. Assumption \ref{a:technology} only requires that
the latent technology $\omega_{t}$ has a one-to-one relation with the ratio $%
\beta_{m^1}(\omega_{t})/\beta_{m^2}(\omega_{t})$ of the elasticities with
respect to two flexible inputs. One restrictive feature of this
non-collinearity assumption is that it rules out constant coefficients ($%
\beta_{m^1}(\omega_{t}) \equiv \beta_{m^1}$ and $\beta_{m^2}(\omega_{t})
\equiv \beta_{m^2}$) for the two flexible inputs.

Finally, we state the following independence assumption for the
idiosyncratic shock $\eta_{t}$, which is standard in the literature. It
requires that the idiosyncratic shock $\eta_{t}$ is unknown by a firm at the
time when it makes input decisions for production to take place in period $t$.

\begin{assumption}[Independence]
\label{a:independence} $(l_{t},k_{t},\omega_{t},p_{t}^{m^1},p_{t}^{m^2})
\perp\!\!\!\perp \eta_{t}$.
\end{assumption}

We introduce the flexible input cost ratio defined by 
\begin{equation}  \label{eq:ratio}
r^{1,2}_{t} = \frac{p_{t}^{m^1} \exp(m^1_{t})}{p_{t}^{m^2} \exp(m^2_{t})}.
\end{equation}
We argue in Lemmas \ref{lemma:r} and \ref{lemma:control_variable} below in
the identification section that this ratio plays an important role as a
control variable for the latent productivity $\omega_{t}$ under Assumptions %
\ref{a:flexible_input} and \ref{a:technology}. We emphasize that the input
prices, $p^{m^1}_{t}$ and $p^{m^2}_{t}$, need not be observed in data. It
suffices to observe the input costs, $p_{t}^{m^1} \exp(m^1_{t})$ and $%
p_{t}^{m^2} \exp(m^2_{t})$.

We also introduce the \textit{ex ante} input cost share of $m^\iota$ for
each $\iota \in \{1,2\}$, defined by 
\begin{equation}  \label{eq:share}
s^\iota_{t} = \frac{ p_{t}^{m^\iota} \exp(m^\iota) }{ \E[ p_{t}^y
\exp(y_{t}) | l_{t}, k_{t}, m^1_{t}, m^2_{t}, r^{1,2}_{t} ] }.
\end{equation}
This is an \textit{ex ante} share because the output is given by the
conditional expectation given the information $(l_{t}, k_{t}, m^1_{t},
m^2_{t}, r^{1,2}_{t})$ observable from the viewpoint of the econometrician
(as well as the firm) before the idiosyncratic shock $\eta_{t}$ is realized.
The prices, $p^{m^1}_{t}$, $p^{m^2}_{t}$ or $p_{t}^y$, need not be observed
in data. It suffices to observe the input costs, $p_{t}^{m^1} \exp(m^1_{t})$
and $p_{t}^{m^2} \exp(m^2_{t})$, and the output value, $p_{t}^y \exp(y_{t})$%
. We argue in Lemma \ref{lemma:beta_m} below in the identification section
that this \textit{ex ante} cost share can be used to identify the
coefficients, $\beta_{m^1}(\omega_{t})$ and $\beta_{m^2}(\omega_{t})$, of
the flexible inputs, $m^1$ and $m^2$, under Assumptions \ref%
{a:flexible_input}, \ref{a:finite}, \ref{a:measurable}, \ref{a:technology},
and \ref{a:independence} .


\section{An Overview}

\label{sec:overview}


\subsection{An Overview of the Construction of a Control Variable}

\label{sec:overview_control_variable}

It is perhaps intuitive and is also formally shown in Lemma \ref{lemma:r}
below that Assumption \ref{a:flexible_input} (flexible input choice) implies
the equality $r^{1,2}_{t} = \beta_{m^1}(\omega_{t}) / \beta_{m^2}(\omega_{t})
$ between the flexible input cost ratio and the ratio of the output
elasticities with respect to the two flexible inputs. By this equality,
econometricians may use the flexible input cost ratio $r^{1,2}_{t}$ in order
to control for the latent technology $\omega_{t}$ under Assumption \ref%
{a:technology}. Figure \ref{fig:overview_control_variable} graphically
describes how this is possible. The base figure used here is copied from
panel (b) of Figure \ref{fig:a:technology}, which illustrates a case where
Assumption \ref{a:technology} is satisfied. Figure \ref%
{fig:overview_control_variable} shows that firms with the flexible input
cost ratio $r^{1,2}_{t} = 0.8$, $1.0$, and $1.2$ are associated with the
latent technological levels of $\omega_{t} = \omega^{\prime }$, $%
\omega^{\prime \prime }$, and $\omega^{\prime \prime \prime }$,
respectively. In this manner, Assumption \ref{a:technology} implies that
controlling for $r^{1,2}_{t}$ is equivalent to controlling for $w_{t}$. This
intuitive illustration also indicates how powerful and important Assumption %
\ref{a:technology} is for our identification strategy.

\begin{figure}[tbp]
\centering
\setlength{\unitlength}{0.6cm} 
\begin{tabular}{c}
\begin{picture}(22,20) \multiput(2,2)(2,0){1}{\vector(0,2){16}}
\multiput(2,2)(2,0){1}{\vector(2,0){16}} \put(18,1){$\beta_{m^1}$}
\put(0,18){$\beta_{m^2}$} \linethickness{0.25mm} \qbezier(2,6)(10,6)(18,12)
\put(13,8.3){$\omega_{t} \mapsto
(\beta_{m^1}(\omega_{t}),\beta_{m^2}(\omega_{t}))$} \linethickness{0.1mm}
\curvedashes[0.5mm]{1,1} \qbezier(2,2)(8,10)(14,18)
\qbezier(2,2)(10,10)(16,16) \qbezier(2,2)(10,8)(18,14)
\put(4.72,6.62){$\omega'$} \put(6.12,6.92){$\omega''$}
\put(8.30,7.49){$\omega'''$} \put(5.14,6.22){\circle*{0.35}}
\put(6.51,6.51){\circle*{0.35}} \put(8.89,7.10){\circle*{0.35}}
\put(14,18){$r^{1,2}_{t} = 0.8$} \put(16,16){$r^{1,2}_{t} = 1.0$}
\put(18,14){$r^{1,2}_{t} = 1.2$} \end{picture}%
\end{tabular}
${}$%
\caption{Illustrations of the construction of the control variable $%
r^{1,2}_{t}$ for the latent technology $\protect\omega_{t}$ under Assumption 
\protect\ref{a:technology}. Firms with the flexible input cost ratio $%
r^{1,2}_{t} = 0.8$, $1.0$, and $1.2$ are associated with the latent
technological levels of $\protect\omega_{t} = \protect\omega^{\prime }$, $%
\protect\omega^{\prime \prime }$, and $\protect\omega^{\prime \prime \prime }
$, respectively.}
\label{fig:overview_control_variable}
\end{figure}
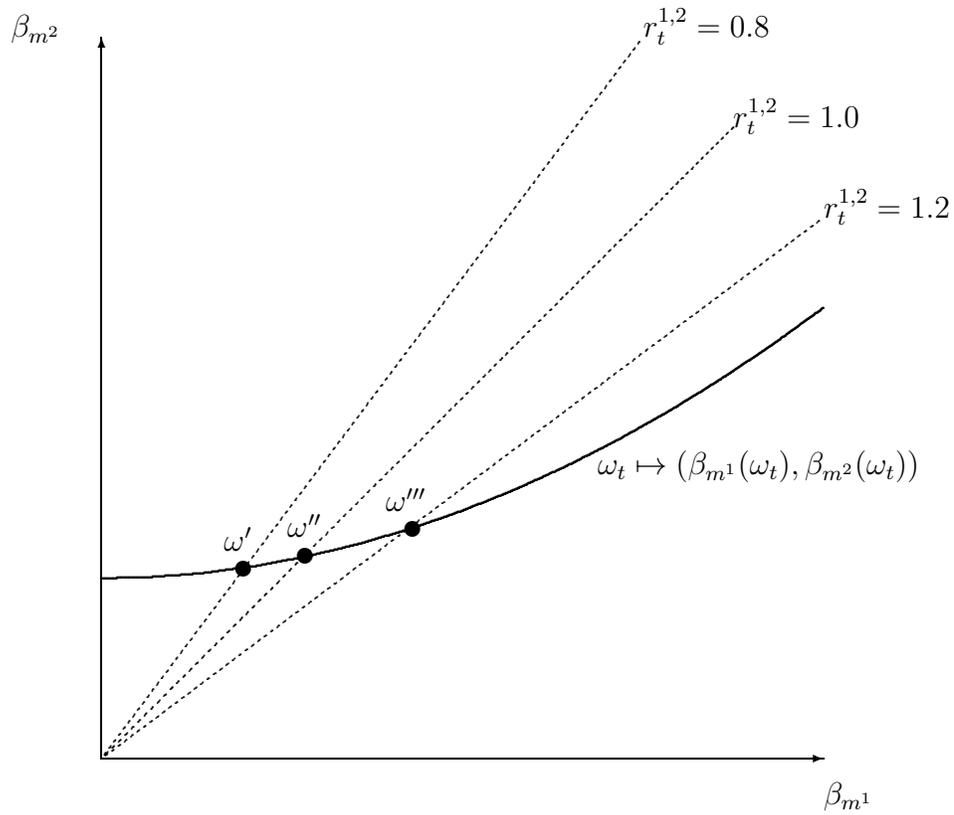


\subsection{An Overview of the Closed-Form Identifying Formulas}

\label{sec:overview_identifying_formulas}

In this section, we present a brief overview of all the closed-form
identifying formulas which we formally derive in Section \ref%
{sec:main_results}. For the model (\ref{eq:production_function})--(\ref%
{eq:short_hand}) equipped with Assumptions \ref{a:flexible_input}, \ref%
{a:finite}, \ref{a:measurable}, \ref{a:technology}, and \ref{a:independence}%
, the heterogeneous coefficients are identified for each firm residing in
the ``locality of identification'' (to be formally defined in Section \ref%
{sec:locality}) with the following closed-form formulas. The coefficients of
flexible inputs, $m^1$ and $m^2$, are identified in closed form by 
\begin{align}
\beta_{m^1}(\omega_{t}) &= s^1_{t} \qquad\text{and}  \notag \\
\beta_{m^2}(\omega_{t}) &= s^2_{t}.  \label{eq:overview_beta_m}
\end{align}
The coefficients of labor $l$ and capital $k$ are identified in turn in
closed form by 
\begin{align}
\beta_{l}(\omega_{t}) &= \left. \frac{\partial}{\partial l} E[ y_t -
\beta_{m^1}(\omega_{t}) m^1_{t} - \beta_{m^2}(\omega_{t}) m^2_{t} \ | \
l_{t} = l, k_{t}, r^{1,2}_{t}] \right\vert_{l=l_{t}} \qquad\text{and}  \notag
\\
\beta_{k}(\omega_{t}) &= \left. \frac{\partial}{\partial k} E[ y_t -
\beta_{m^1}(\omega_{t}) m^1_{t} - \beta_{m^2}(\omega_{t}) m^2_{t} \ | \
l_{t}, k_{t} = k, r^{1,2}_{t}] \right\vert_{k=k_{t}}.
\label{eq:overview_beta_l_beta_k}
\end{align}
Finally, the additive productivity is identified in closed form by 
\begin{equation}  \label{eq:overview_beta_0}
\beta_{0}(\omega_{t}) = E[ y_t - \beta_{l}(\omega_{t}) l_{t} -
\beta_{k}(\omega_{t}) k_{t} - \beta_{m^1}(\omega_{t}) m^1_{t} -
\beta_{m^2}(\omega_{t}) m^2_{t} \ | \ l_{t}, k_{t}, r^{1,2}_{t}].
\end{equation}
The next section presents a formal argument to derive these closed-form
identifying formulas.


\section{Main Results}

\label{sec:main_results}


\subsection{Locality of Identification}

\label{sec:locality}

A potential obstacle to identification of production function models is that
functional dependence of input choices on the latent technology may induce
rank deficiency and thus a failure in identification \citep{AcCaFr15}. A few
structural models to eliminate this functional dependence problem are
proposed by \citet{AcCaFr15}. Among those structural models, \cite{HuHuSa17}
empirically support the second model of \citet{AcCaFr15} that labor input $%
l_{t}$ as well as capital $k_{t}$ are chosen earlier than period $t$. This
structural assumption is also consistent with our Assumption \ref%
{a:flexible_input} where only $m_{t}^{1}$ and $m_{t}^{2}$ are treated as
flexible inputs. If $l_{t}$ and $k_{t}$ are chosen prior to realization of $%
\omega _{t}$, then stochastic evolution of $\omega _{t}$ will allow for
elimination of functional dependence in the sense that it allows for a
non-degenerate conditional distribution of $l_{t}$ given $(k_{t},\omega _{t})
$ and a non-degenerate conditional distribution of $k_{t}$ given $%
(l_{t},\omega _{t})$. As we discuss below after Lemma \ref%
{lemma:beta_l_beta_k}, this non-degeneracy is crucial for identification of $%
\beta _{l}(\cdot )$ and $\beta _{k}(\cdot )$ in the relevant locality. In
light of the timing assumption with stochastic evolution of $\omega _{t}$,
we assume that $l_{t}$, $k_{t}$, and $\omega _{t}$ are continuous random
variables throughout this paper, and thus the location $(l_{t},k_{t},\omega
_{t})=(l^{\ast },k^{\ast },\omega ^{\ast })$ defined below enables the local
identification.

\begin{definition}[Locality of Identification]
\label{def:locality_of_identification} The point $(l_{t},k_{t},\omega_{t}) =
(l^\ast,k^\ast,\omega^\ast)$ is called the locality of identification if
there are real numbers $\underline{f}, \overline{f} \in (0,\infty)$ such
that $\underline{f} < f_{l_{t},k_{t},\omega_{t}}(l,k,\omega) < \overline{f}$
for all $(l,k,\omega)$ in a neighborhood of $(l^\ast,k^\ast,\omega^\ast)$.
\end{definition}

\begin{proposition}[Functional Independence]
\label{prop:functional_independence} If $(l^\ast,k^\ast,\omega^\ast)$ is the
locality of identification according to Definition \ref%
{def:locality_of_identification}, then: (i) the conditional density function 
$f_{l_{t}|k_{t},\omega_{t}}( \ \cdot \ |k^\ast,\omega^\ast)$ exists and is
positive in a neighborhood of $l=l^\ast$; and (ii) the conditional density
function $f_{k_{t}|l_{t},\omega_{t}}( \ \cdot \ |l^\ast,\omega^\ast)$ exists
and is positive in a neighborhood of $k=k^\ast$.
\end{proposition}

\begin{proof}
By Definition \ref{def:locality_of_identification}, there is $\varepsilon \in (0,\infty)$ such that $\underline{f} < f_{l_{t},k_{t},\omega_{t}}(l,k,t) < \overline{f}$ for all $(l,k,t)$ in an $\varepsilon$-ball of $(l^\ast,k^\ast,\omega^\ast)$.
Thus, $f_{k_{t},\omega_{t}}(k^\ast,\omega^\ast) > 2\varepsilon\underline{f} > 0$.
Therefore, the conditional density function $f_{l_{t}|k_{t},\omega_{t}}( \ \cdot \ |k^\ast,\omega^\ast) = f_{l_{t},k_{t},\omega_{t}}( \ \cdot \ ,k^\ast,\omega^\ast) / f_{k_{t},\omega_{t}}(k^\ast,\omega^\ast)$ exists, and is bounded by $\overline{f} / f_{k_{t},\omega_{t}}(k^\ast,\omega^\ast)$ in the $\varepsilon$-ball of $l=l^\ast$.
This proves part (i).
A proof of part (ii) similar by exchanging the roles of $l$ and $k$.
\end{proof}


\subsection{Control Variable}

\label{sec:control_variable}

We construct a control variable via the first-order condition explicitly
exploiting the structural information. The flexible input choice rule in
Assumption \ref{a:flexible_input} yields the following restrictions as the
first-order condition. 
\begin{align}  \label{eq:foc1}
p_{t}^y \beta_{m^\iota}(\omega_{t}) \exp\left(\Psi\left(
l_{t},k_{t},m^1_{t},m^2_{t},\omega_{t} \right)\right) \E\left[\exp(\eta_{t})%
\right] = p_{t}^{m^\iota} \exp(m^\iota)
\end{align}
for each $\iota \in \{1,2\}$. Furthermore, Assumption \ref{a:finite}
guarantees that the solution to the flexible input choice problem exists,
and is explicitly given by 
\begin{align}
m^1_{t} = \ln p_{t}^y &+ \frac{(1-\beta_{m^2}(\omega_{t}))\ln\frac{%
\beta_{m^1}(\omega_{t})}{p_{t}^{m^1}} + \beta_{m^2}(\omega_{t})\ln\frac{%
\beta_{m^2}(\omega_{t})}{p_{t}^{m^2}}}{1-\beta_{m^1}(\omega_{t})-%
\beta_{m^2}(\omega_{t})}  \notag \\
&+ \frac{\beta_{l}(\omega) l_{t} + \beta_{k}(\omega) k_{t} +
\beta_{0}(\omega_{t}) + \ln\E[\exp(\eta_{t})]}{1-\beta_{m^1}(\omega_{t})-%
\beta_{m^2}(\omega_{t})}  \label{eq:m1}
\end{align}
\begin{align}
m^2_{t} = \ln p_{t}^y &+ \frac{\beta_{m^1}(\omega_{t})\ln\frac{%
\beta_{m^1}(\omega_{t})}{p_{t}^{m^1}} + (1-\beta_{m^1}(\omega_{t}))\ln\frac{%
\beta_{m^2}(\omega_{t})}{p_{t}^{m^2}}}{1-\beta_{m^1}(\omega_{t})-%
\beta_{m^2}(\omega_{t})}  \notag \\
&+ \frac{\beta_{l}(\omega) l_{t} + \beta_{k}(\omega) k_{t} +
\beta_{0}(\omega_{t}) + \ln\E[\exp(\eta_{t})]}{1-\beta_{m^1}(\omega_{t})-%
\beta_{m^2}(\omega_{t})}  \label{eq:m2}
\end{align}
These two equations under Assumption \ref{a:measurable} imply that $%
(m^1_{t},m^2_{t})$ is a measurable function of the state variables $%
(l_{t},k_{t},\omega_{t},p_{t}^{m^1},p_{t}^{m^2})$. With these solutions to
the static optimization problem as auxiliary tools, we now proceed with the
construction of a control variable for the unobserved technology $\omega_{t}$%
. The next lemma shows that the flexible input cost ratio $r^{1,2}_{t}$
defined in (\ref{eq:ratio}) identifies the ratio of the heterogeneous
coefficients of two flexible inputs.

\begin{lemma}[Identification of the Ratio $\protect\beta_{m^1}(\protect\omega%
_{t}) / \protect\beta_{m^2}(\protect\omega_{t})$]
\label{lemma:r} If Assumption \ref{a:flexible_input} is satisfied, then 
\begin{equation}  \label{eq:identification_ratio}
\frac{\beta_{m^1}(\omega_{t})}{\beta_{m^2}(\omega_{t})} = r^{1,2}_{t}.
\end{equation}
\end{lemma}

\begin{proof}
Assumption \ref{a:flexible_input} yields the first-order condition (\ref{eq:foc1}).
Taking the ratio of this first-order condition for $\iota=1$ over the first-order condition for $\iota=2$ yields (\ref{eq:identification_ratio}).
\end{proof}

From Assumption \ref{a:technology} and (\ref{eq:identification_ratio}) that
holds under Assumption \ref{a:flexible_input}, we can see that the flexible
input cost ratio $r^{1,2}_{t}$ defined in (\ref{eq:ratio}) can be used as a
control variable for the unobserved latent productivity $\omega_{t}$.
Specifically, we state the following lemma.

\begin{lemma}[Control Variable]
\label{lemma:control_variable} If Assumptions \ref{a:flexible_input} and \ref%
{a:technology} satisfied, then there exists a measurable invertible function 
$\phi$ with a measurable inverse $\phi^{-1}$ such that 
\begin{equation}  \label{eq:control_variable}
\omega_{t} = \phi(r^{1,2}_{t}).
\end{equation}
\end{lemma}

\begin{proof}
The claim follows from Lemma \ref{lemma:r} and Assumption \ref{a:technology}.
\end{proof}

In light of this lemma, we can interpret the flexible input cost ratio $%
r_{t}^{1,2}$ as a normalized observable measure of the unobserved latent
productivity $\omega _{t}$ through the normalizing transformation $\phi $.
With this interpretation, the identification of $\beta _{l}(\cdot )$, $\beta
_{k}(\cdot )$, $\beta _{m^{1}}(\cdot )$, $\beta _{m^{2}}(\cdot )$, and $%
\beta _{0}(\cdot )$ may be achieved by the identification of $\beta
_{l}\circ \phi (\cdot )$, $\beta _{k}\circ \phi (\cdot )$, $\beta
_{m^{1}}\circ \phi (\cdot )$, $\beta _{m^{2}}\circ \phi (\cdot )$, and $%
\beta _{0}\circ \phi (\cdot )$, respectively, which take the normalized
observable measure $r_{t}^{1,2}$ of the unobserved productivity $\omega _{t}$%
. The closed-form identification results stated as Lemmas \ref{lemma:beta_m}%
, \ref{lemma:beta_l_beta_k}, and \ref{lemma:beta_0} in Section \ref%
{sec:identification} indeed consist of identifying formulas in the forms of
these function compositions.


\subsection{Identification}

\label{sec:identification}


\subsubsection{Coefficients of Flexible Inputs}

Recall the \textit{ex ante} input cost shares, $s^1_{t}$ and $s^2_{t}$, of $%
m^1$ and $m^2$, respectively, defined in (\ref{eq:share}). These \textit{ex
ante} input cost shares are not directly observable from data, but are
directly identifiable from data. The following lemma shows that the marginal
product $\beta_{m^\iota}(\omega_{t})$ of flexible input $m^\iota$ can be
identified by this \textit{ex ante} input cost share for each $\iota \in
\{1,2\}$.

\begin{lemma}[Identification of $\protect\beta_{m^1}(\cdot)$ and $\protect%
\beta_{m^2}(\cdot)$]
\label{lemma:beta_m} If Assumptions \ref{a:flexible_input}, \ref{a:finite}, %
\ref{a:measurable}, \ref{a:technology}, and \ref{a:independence} are
satisfied, then 
\begin{align}
\beta_{m^1}(\omega_{t}) &= \beta_{m^1}(\phi(r^{1,2}_{t})) = s^1_{t} \qquad%
\text{and}  \notag \\
\beta_{m^2}(\omega_{t}) &= \beta_{m^2}(\phi(r^{1,2}_{t})) = s^2_{t}
\label{eq:identification_beta_m}
\end{align}
hold.
\end{lemma}

This lemma establishes the identifying formulas (\ref{eq:overview_beta_m})
presented in Section \ref{sec:overview}.

\begin{proof}
The proof of this lemma consists of four steps.
First, note that Lemma \ref{lemma:control_variable} under Assumptions \ref{a:flexible_input} and \ref{a:technology} implies the equaivalence between the two sigma algebras:
\begin{equation}\label{eq:same_sigma_algebra}
\sigma( l_{t}, k_{t}, m^1_{t}, m^2_{t}, r^{1,2}_{t} )
= 
\sigma( l_{t}, k_{t}, m^1_{t}, m^2_{t}, \omega_{t} ).
\end{equation}

Second, Assumptions \ref{a:flexible_input} and \ref{a:finite} yield (\ref{eq:m1}) and (\ref{eq:m2}), which in turn imply under Assumption \ref{a:measurable} that $(m^1_{t},m^2_{t})$ is a measurable function of $(l_{t},k_{t},\omega_{t},p_{t}^{m^1},p_{t}^{m^2})$.
In light of this, applying Theorem 2.1.6 of \citet{Durrett2010} to Assumption \ref{a:independence} yields\footnote{
The referenced theorem says that $X \ind Y$ implies $f_1(X) \ind f_2(X)$ for any measurable functions $f_1$ and $f_2$.
}
\begin{equation}\label{eq:independence_m}
(l_{t},k_{t},m^1_{t},m^2_{t},\omega_{t}) \ind \eta_{t}.
\end{equation}

Third, we obtain the following chain of equalities.
\begin{align}
\E[ \exp(y_{t}) | l_{t}, k_{t}, m^1_{t}, m^2_{t}, r^{1,2}_{t} ]
=
&\E[ \exp(y_{t}) | l_{t}, k_{t}, m^1_{t}, m^2_{t}, \omega_{t} ]
\nonumber\\
=
&\E[ \exp\left(\Psi\left( l_{t},k_{t},m^1_{t},m^2_{t},\omega_{t} \right) + \eta_{t}\right) | l_{t}, k_{t}, m^1_{t}, m^2_{t}, \omega_{t} ]
\nonumber\\
= 
&\exp\left(\Psi\left( l_{t},k_{t},m^1_{t},m^2_{t},\omega_{t} \right)\right)
\E[\exp(\eta_{t}) | l_{t}, k_{t}, m^1_{t}, m^2_{t}, \omega_{t}] 
\nonumber\\
= 
&\exp\left(\Psi\left( l_{t},k_{t},m^1_{t},m^2_{t},\omega_{t} \right)\right)
\E[\exp(\eta_{t})] 
\label{eq:conditional_expectation_y}
\end{align}
where the first equality is due to (\ref{eq:same_sigma_algebra}), the second equality follows from substitution of the gross-output production function (\ref{eq:production_function}), the third equality follows from the property of the conditional expectation that $\E[\Psi_1(X_1)\Psi_2(X_2)|X_1]=\Psi_1(X_1)\E[\Psi_2(X_2)|X_1]$, and the fourth equality follows from (\ref{eq:independence_m}).

Finally, taking the ratio of (\ref{eq:foc1}) to (\ref{eq:conditional_expectation_y}) yields
$$
\beta_{m^\iota}(\omega_{t})
= 
\frac{ p_{t}^{m^\iota} \exp(m^\iota) }{ p_{t}^y \E[ \exp(y_{t}) | l_{t}, k_{t}, m^1_{t}, m^2_{t}, r^{1,2}_{t} ] }
$$
for each $\iota \in \{1,2\}$. 
By the definition of $s^\iota_{t}$ given in (\ref{eq:share}), this proves (\ref{eq:identification_beta_m}).
\end{proof}


\subsubsection{Coefficients of Labor and Capital Inputs}

We now introduce the short-hand notation, 
\begin{equation*}
\tilde y_{t} := y_{t} - s^1_{t} m^1_{t} - s^2_{t} m^2_{t}, 
\end{equation*}
which can be interpreted as the output net of the flexible input
contributions by Lemma \ref{lemma:beta_m} under Assumptions \ref%
{a:flexible_input}, \ref{a:finite}, \ref{a:measurable}, \ref{a:technology},
and \ref{a:independence}. Thus, we can rewrite the gross-output production
function (\ref{eq:production_function})--(\ref{eq:short_hand}) into the
net-output production function 
\begin{equation}  \label{eq:net_production_function}
\tilde y_{t} = \beta_{l}(\omega_{t}) l_{t} + \beta_{k}(\omega_{t}) k_{t} +
\beta_{0}(\omega_{t}) + \eta_{t}.
\end{equation}
It remains to identify the remaining heterogeneous coefficient functions $%
\beta_{l}(\cdot)$, $\beta_{k}(\cdot)$, and $\beta_{0}(\cdot)$. The next
lemma proposes the identification of $\beta_{l}(\cdot)$ and $\beta_{k}(\cdot)
$.

\begin{lemma}[Identification of $\protect\beta_{l}(\cdot)$ and $\protect\beta%
_{k}(\cdot)$]
\label{lemma:beta_l_beta_k} If Assumptions \ref{a:flexible_input}, \ref%
{a:finite}, \ref{a:measurable}, \ref{a:technology}, and \ref{a:independence}
are satisfied, then 
\begin{align}
\beta_{l}(\omega_{t}) &= \beta_{l}(\phi(r^{1,2}_{t})) = \left.\frac{\partial%
}{\partial l}\E[\tilde y_{t} | l_{t}=l, k_{t}, r^{1,2}_{t}]%
\right\vert_{l=l_{t}} \quad\text{and}  \notag \\
\beta_{k}(\omega_{t}) &= \beta_{k}(\phi(r^{1,2}_{t})) = \left.\frac{\partial%
}{\partial k}\E[\tilde y_{t} | l_{t}, k_{t}=k, r^{1,2}_{t}]%
\right\vert_{k=k_{t}}  \label{eq:beta_l_beta_k}
\end{align}
hold.
\end{lemma}

This lemma establishes the identifying formulas (\ref%
{eq:overview_beta_l_beta_k}) presented in Section \ref{sec:overview}.

\begin{proof}
The proof of this lemma consists of four steps.
First, note that Assumption \ref{a:technology} together with Lemma \ref{lemma:r} under Assumption \ref{a:flexible_input} implies the equivalence between the two sigma algebras:
\begin{equation}\label{eq:same_sigma_algebra_without_m}
\sigma( l_{t}, k_{t}, r^{1,2}_{t} )
= 
\sigma( l_{t}, k_{t}, \omega_{t} ).
\end{equation}

Second, applying the decomposition property of the semi-graphoid axiom \citep[][pp. 11]{Pearl2000} to Assumption \ref{a:independence} yields
$(l_{t},k_{t},\omega_{t}) \ind \eta_{t}$.
This independence and the restriction $\E[\eta_{t}]$ in the gross-output production function model (\ref{eq:production_function}) together yield
\begin{equation}\label{eq:eta_l_k_omega}
\E[ \eta_{t} | l_{t}, k_{t}, \omega_{t}] = 0.
\end{equation}

Third, we obtain the following chain of equalities.
\begin{align}
\E[\tilde y_{t} | l_{t}, k_{t}, r^{1,2}_{t}]
&=
\E[\tilde y_{t} | l_{t}, k_{t}, \omega_{t}]
\nonumber\\
&=
\E[ \beta_{l}(\omega_{t}) l_{t} + \beta_{k}(\omega_{t}) k_{t} + \beta_{0}(\omega_{t}) + \eta_{t} | l_{t}, k_{t}, \omega_{t}]
\nonumber\\
&=
\beta_{l}(\omega_{t}) l_{t} + \beta_{k}(\omega_{t}) k_{t} + \beta_{0}(\omega_{t}) + \E[ \eta_{t} | l_{t}, k_{t}, \omega_{t}]
\nonumber\\
&=
\beta_{l}(\omega_{t}) l_{t} + \beta_{k}(\omega_{t}) k_{t} + \beta_{0}(\omega_{t})
\label{eq:conditional_expectation_y_tilde}
\end{align}
where the first equality is due to (\ref{eq:same_sigma_algebra_without_m}), the second equality follows from a substitution of the net-output production function (\ref{eq:net_production_function}) which is valid by Lemma \ref{lemma:beta_m} under Assumptions \ref{a:flexible_input}, \ref{a:finite}, \ref{a:measurable}, \ref{a:technology}, and \ref{a:independence}, the third equality follows from the information $\sigma( l_{t}, k_{t}, \omega_{t} )$ on which the conditional expectation is taken, and the fourth equality follows from (\ref{eq:eta_l_k_omega}).

Fourth, substituting (\ref{eq:control_variable}) under Assumptions \ref{a:flexible_input} and \ref{a:technology} in (\ref{eq:conditional_expectation_y_tilde}), we obtain
$$
\E[\tilde y_{t} | l_{t}=l, k_{t}=k, r^{1,2}_{t}=r] = \beta_{l}(\phi(r)) l + \beta_{k}(\phi(r)) k + \beta_{0}(\phi(r)).
$$
The map $(l,k,r) \mapsto \E[\tilde y_{t} | l_{t}=l, k_{t}=k, r^{1,2}_{t}=r]$ is thus an affine function of $(l,k)$, and thus is differentiable in $(l,k)$ in particular.
Differentiating both sides of this equation with respect to $l$ and $k$ at $(l,k,r) = (l_{t},k_{t},r^{1,2}_{t})$ yields
\begin{align*}
\left.\frac{\partial}{\partial l}\E[\tilde y_{t} | l_{t}=l, k_{t}=k, r^{1,2}_{t}=r]\right\vert_{(l,k,r)=(l_{t},k_{t},r^{1,2}_{t})} &= \beta_{l}(\phi(r^{1,2}_{t})) = \beta_{l}(\omega_{t})
\qquad\text{and}\\
\left.\frac{\partial}{\partial k}\E[\tilde y_{t} | l_{t}=l, k_{t}=k, r^{1,2}_{t}=r]\right\vert_{(l,k,r)=(l_{t},k_{t},r^{1,2}_{t})} &= \beta_{k}(\phi(r^{1,2}_{t})) = \beta_{k}(\omega_{t}), 
\end{align*}
respectively.
This shows (\ref{eq:beta_l_beta_k}).
\end{proof}

Lemma \ref{lemma:beta_l_beta_k} paves the way for identification of $%
\beta_{l}(\omega_{t})$ and $\beta_{k}(\omega_{t})$, but it in fact does not
guarantee the identification by itself. To make sense of (\ref%
{eq:beta_l_beta_k}) as identifying formulas, $l_{t}$ should be functionally
independent of $(k_{t},r^{1,2}_{t})$, and, similarly, $k_{t}$ should be
functionally independent of $(l_{t},r^{1,2}_{t})$ in the language of \cite%
{AcCaFr15}. Proposition \ref{prop:functional_independence} shows that the
locality of identification introduced in Definition \ref%
{def:locality_of_identification} satisfies the functional independence.
Therefore, (\ref{eq:beta_l_beta_k}) can be interpreted as the identifying
formulas at such localities.


\subsubsection{Additive Technology}

We now introduce the further short-hand notation, 
\begin{equation}  \label{eq:y_tilde_tilde}
\tilde{\tilde y}_{t} := \tilde y_{t} - \left.\frac{\partial}{\partial l}%
\E[\tilde y_{t} | l_{t}=l, k_{t}, r^{1,2}_{t}]\right\vert_{l=l_{t}} l_{t} -
\left.\frac{\partial}{\partial k}\E[\tilde y_{t} | l_{t}, k_{t}=k,
r^{1,2}_{t}]\right\vert_{k=k_{t}} k_{t},
\end{equation}
which can be interpreted as the residual of the production function (\ref%
{eq:production_function}) by Lemma \ref{lemma:beta_m} and \ref%
{lemma:beta_l_beta_k} under Assumptions \ref{a:flexible_input}, \ref%
{a:finite}, \ref{a:measurable}, \ref{a:technology}, and \ref{a:independence}%
. With this new notation, therefore, we can rewrite the net-output
production function (\ref{eq:net_production_function}) into 
\begin{equation*}
\tilde{\tilde y}_{t} = \beta_{0}(\omega_{t}) + \eta_{t}.
\end{equation*}
The final step is to identify the additive productivity $\beta_{0}(%
\omega_{t})$. The next lemma provides the identification of $%
\beta_{0}(\omega_{t})$.

\begin{lemma}[Identification of $\protect\beta_{0}(\cdot)$]
\label{lemma:beta_0} If Assumptions \ref{a:flexible_input}, \ref{a:finite}, %
\ref{a:measurable}, \ref{a:technology}, and \ref{a:independence} are
satisfied, then 
\begin{align}
\beta_{0}(\omega_{t}) = \beta_{0}(\phi(r^{1,2}_{t})) = \E\left[\tilde{\tilde
y}_{t} \left\vert l_{t}, k_{t}, r^{1,2}_{t}\right.\right]
\end{align}
holds.
\end{lemma}

This lemma establishes the identifying formula (\ref{eq:overview_beta_0})
presented in Section \ref{sec:overview}.

\begin{proof}
Equation (\ref{eq:conditional_expectation_y_tilde}) in the proof of Lemma \ref{lemma:beta_l_beta_k} can be rewritten as
\begin{align*}
\E\left[\tilde y_{t} - \beta_l(\omega_{t}) l_{t} - \beta_k(\omega_{t}) \left\vert l_{t}, k_{t}, r^{1,2}_{t} \right.\right]
=
\beta_{0}(\omega_{t}).
\end{align*}
Substituting the definition (\ref{eq:y_tilde_tilde}) of $\tilde{\tilde y}_{t}$ together with (\ref{eq:beta_l_beta_k}) of Lemma \ref{lemma:beta_l_beta_k} in the above equation proves the corollary.
\end{proof}


\subsection{Summary of the Main Results}

Summarizing the identification steps stated as Proposition \ref%
{prop:functional_independence} and Lemmas \ref{lemma:r}, \ref%
{lemma:control_variable}, \ref{lemma:beta_m}, \ref{lemma:beta_l_beta_k}, and %
\ref{lemma:beta_0}, we obtain the following theorem.

\begin{theorem}[Identification]
Suppose that Assumptions \ref{a:flexible_input}, \ref{a:finite}, \ref%
{a:measurable}, \ref{a:technology}, and \ref{a:independence} are satisfied
for the model (\ref{eq:production_function})--(\ref{eq:short_hand}). The
parameter vector $(\beta_{l}(\omega_{t}),\beta_{k}(\omega_{t}),\beta_{m^1}(%
\omega_{t}),\beta_{m^2}(\omega_{t}),\beta_{0}(\omega_{t}))$ with the latent
productivity $\omega_t$ is identified if there exists $(l_t,k_t)$ such that $%
(l_t,k_t,\omega_t)$ is at a locality of identification (Definition \ref%
{def:locality_of_identification}).
\end{theorem}

For a summary of all the closed-form identifying formulas, we refer the
readers the overview in Section \ref{sec:overview_identifying_formulas}. In
this paper, we focus on the identification problem, and leave aside methods
of estimation. Since what we identify are functions, $\beta _{l}\circ \phi
(\cdot )$, $\beta _{k}\circ \phi (\cdot )$, $\beta _{m^{1}}\circ \phi (\cdot
)$, $\beta _{m^{2}}\circ \phi (\cdot )$, and $\beta _{0}\circ \phi (\cdot )$%
, explicitly expressed in terms of nonparametric conditional expectation
functions and their derivatives, one may use analog nonparametric estimation
methods \citep[e.g.,][]{Ch07} to obtain function estimates $\widehat{\beta
_{l}\circ \phi }(\cdot )$, $\widehat{\beta _{k}\circ \phi }(\cdot )$, $%
\widehat{\beta _{m^{1}}\circ \phi }(\cdot )$, $\widehat{\beta _{m^{2}}\circ
\phi }(\cdot )$, and $\widehat{\beta _{0}\circ \phi }(\cdot )$. The
estimates of heterogeneous coefficients are then computed by $\widehat{\beta
_{l}}(\omega _{t})=\widehat{\beta _{l}\circ \phi }(r_{t}^{1,2})$, $\widehat{%
\beta _{k}}(\omega _{t})=\widehat{\beta _{k}\circ \phi }(r_{t}^{1,2})$, $%
\widehat{\beta _{m^{1}}}(\omega _{t})=\widehat{\beta _{m^{1}}\circ \phi }%
(r_{t}^{1,2})$, $\widehat{\beta _{m^{2}}}(\omega _{t})=\widehat{\beta
_{m^{2}}\circ \phi }(r_{t}^{1,2})$, and $\widehat{\beta _{0}}(\omega _{t})=%
\widehat{\beta _{0}\circ \phi }(r_{t}^{1,2})$.


\section{Alternative Models}

\label{sec:alternative_models}

In the baseline model, we considered a parsimonious form that consists of
only the two observed state variables $(l_{t},k_{t})$ and the two observed
flexible input variables $(m^1_{t},m^2_{t})$. In this section, we remark
that it is possible to augment the vectors of the observed state variables $%
(l_{t},k_{t})$ and/or the observed flexible input variables $%
(m^1_{t},m^2_{t})$ -- see Sections \ref{sec:more_state_variables}--\ref%
{sec:more_flexible_input_variables}. On the other hand, it is also possible
to reduce the model with just one $m_{t}$ variable, provided that $l_{t}$
satisfies the timing assumption as a flexible input -- see Section \ref%
{sec:single_m_variable}.


\subsection{More State Variables}

\label{sec:more_state_variables}

The baseline model treats the labor input $l_{t}$ only of a single type. In
applications, however, researchers often distinguish skilled labor input $%
l^s_{t}$ and unskilled labor input $l^u_{t}$ \citep[e.g.,][]{LePe03}. To
accommodate this distinction, we can augment the gross-output production
function in the logarithm (\ref{eq:production_function}) as 
\begin{equation}  \label{eq:aug:production_function}
y_{t} = \Psi\left( l^s_{t},l^u_{t},k_{t},m^1_{t},m^2_{t},\omega_{t} \right)
+ \eta_{t} \qquad \E[\eta_t]=0,
\end{equation}
where $l^s_{t}$ is the logarithm of skilled labor input, $l^y_{t}$ is the
logarithm of unskilled labor input, and all the other variables are the same
as in the baseline model. Accordingly, the Cobb-Douglas form (\ref%
{eq:short_hand}) is augmented as 
\begin{equation}  \label{eq:aug:short_hand}
\Psi\left(l^s_t,l^u_t,k,m^1_t,m^2_t,\omega_t\right) = \beta_{l^s}(\omega_t)
l^s_t + \beta_{l^u}(\omega_t) l^u_t + \beta_{k}(\omega_t) k_t +
\beta_{m^1}(\omega_t) m^1_t + \beta_{m^2}(\omega_t) m^2_t +
\beta_{0}(\omega_t)
\end{equation}
with heterogeneous coefficients $(\beta_{l^s}(\omega_t),\beta_{l^u}(%
\omega_t),\beta_{k}(\omega_t),\beta_{m^1}(\omega_t),\beta_{m^2}(\omega_t),%
\beta_{0}(\omega_t))$. With the following modifications of Assumptions \ref%
{a:flexible_input}, \ref{a:measurable}, and \ref{a:independence} adapted to
the current augmented model, we can construct the identification results in
similar lines of argument to those we had for the baseline model.

\newtheorem*{assumption_aug_flexible_input}{Assumption
\ref{a:flexible_input}$'$} 
\begin{assumption_aug_flexible_input}[Flexible Input Choice]
A firm at time $t$ with the state variables $(l^s_{t},l^u_{t},k_{t},\omega_{t})$ chooses the flexible input vector $(m^1_{t},m^2_{t})$ by the optimization problem
\begin{eqnarray*}
\max_{(m^1,m^2) \in \mathbb{R}_+^2}
p_{t}^y \exp\left(\Psi\left( l^s_{t},l^u_{t},k_{t},m^1,m^2,\omega_{t} \right)\right)
E[\exp\left(\eta_{t}\right)]
-p_{t}^{m^1} \exp\left( m^1 \right)
-p_{t}^{m^2} \exp\left( m^2 \right),
\end{eqnarray*}
 where $p_{t}^{m^1} > 0$ and $p_{t}^{m^2} > 0$ almost surely.
\end{assumption_aug_flexible_input}

\newtheorem*{assumption_aug_measurable}{Assumption \ref{a:measurable}$'$} 
\begin{assumption_aug_measurable}[Measurability]
$\beta_{l^s}(\cdot)$, $\beta_{l^u}(\cdot)$, $\beta_{k}(\cdot)$, $\beta_{m^1}(\cdot)$, $\beta_{m^2}(\cdot)$ and $\beta_{0}(\cdot)$ are measurable functions.
\end{assumption_aug_measurable}

\newtheorem*{assumption_aug_independence}{Assumption \ref{a:independence}$'$}
\begin{assumption_aug_independence}[Independence]
$(l^s_{t},l^u_{t},k_{t},\omega_{t},p_{t}^{m^1},p_{t}^{m^2}) \ind \eta_{t}$.
\end{assumption_aug_independence}

Assumption \ref{a:flexible_input}$^{\prime }$ proposes that both $l^s_{t}$
and $l^u_{t}$ are predetermined, which is empirically supported by \cite%
{HuHuSa17}. Assumptions \ref{a:measurable}$^{\prime }$ and \ref%
{a:independence}$^{\prime }$ are straightforward extensions of Assumptions %
\ref{a:measurable} and \ref{a:independence}, respectively, suitable for the
current model (\ref{eq:aug:production_function})--(\ref{eq:aug:short_hand}).
With the following modification to the definition of the locality of
identification adapted to the current augmented model, we state as a theorem
the extended identification result.

\newtheorem*{definition_aug_locality_of_identification}{Definition
\ref{def:locality_of_identification}$'$} 
\begin{definition_aug_locality_of_identification}[Locality of Identification]
The point $(l^s_{t},l^u_{t},k_{t},\omega_{t}) = (l^{s\ast},l^{u\ast},k^\ast,\omega^\ast)$ is called the locality of identification if there are real numbers $\underline{f}, \overline{f} \in (0,\infty)$ such that $\underline{f} < f_{l^s_{t},l^u_{t},k_{t},\omega_{t}}(l^s,l^u,k,\omega) < \overline{f}$ for all $(l^s,l^u,k,\omega)$ in a neighborhood of $(l^{s\ast},l^{u\ast},k^\ast,\omega^\ast)$.
\end{definition_aug_locality_of_identification}

\begin{theorem}[Identification]
Suppose that Assumptions \ref{a:flexible_input}$^{\prime }$, \ref{a:finite}, %
\ref{a:measurable}$^{\prime }$, \ref{a:technology}, and \ref{a:independence}$%
^{\prime }$ are satisfied for the model (\ref{eq:aug:production_function})--(%
\ref{eq:aug:short_hand}). The parameter vector $(\beta_{l^s}(\omega_{t}),%
\beta_{l^u}(\omega_{t}),\beta_{k}(\omega_{t}),\beta_{m^1}(\omega_{t}),%
\beta_{m^2}(\omega_{t}),\beta_{0}(\omega_{t}))$ with the latent productivity 
$\omega_t$ is identified if there exists $(l^s_t,l^u_t,k_t)$ such that $%
(l^s_t,l^u_t,k_t,\omega_t)$ is at a locality of identification (Definition %
\ref{def:locality_of_identification}$^{\prime }$).
\end{theorem}

The coefficients of flexible inputs, $m^1$ and $m^2$, are identified in
closed form by 
\begin{align*}
\beta_{m^1}(\omega_{t}) &= s^1_{t} = \frac{ p_{t}^{m^1} \exp(m^1) }{ \E[
p_{t}^y \exp(y_{t}) | l^s_{t}, l^u_{t}, k_{t}, m^1_{t}, m^2_{t}, r^{1,2}_{t}
] } \qquad\text{and} \\
\beta_{m^2}(\omega_{t}) &= s^2_{t} = \frac{ p_{t}^{m^2} \exp(m^2) }{ \E[
p_{t}^y \exp(y_{t}) | l^s_{t}, l^u_{t}, k_{t}, m^1_{t}, m^2_{t}, r^{1,2}_{t}
] }
\end{align*}
The coefficients of skilled labor $l^s$, unskilled labor $l^u$, and capital $%
k$ are identified in turn in closed form by 
\begin{align*}
\beta_{l^s}(\omega_{t}) &= \left. \frac{\partial}{\partial l^s} E[ y_t -
\beta_{m^1}(\omega_{t}) m^1_{t} - \beta_{m^2}(\omega_{t}) m^2_{t} \ | \
l^s_{t} = l^s, l^u_{t}, k_{t}, r^{1,2}_{t}] \right\vert_{l^s=l^s_{t}} \\
\beta_{l^u}(\omega_{t}) &= \left. \frac{\partial}{\partial l^u} E[ y_t -
\beta_{m^1}(\omega_{t}) m^1_{t} - \beta_{m^2}(\omega_{t}) m^2_{t} \ | \
l^s_{t}, l^u_{t} = l, k_{t}, r^{1,2}_{t}] \right\vert_{l^u=l^u_{t}} \qquad%
\text{and} \\
\beta_{k}(\omega_{t}) &= \left. \frac{\partial}{\partial k} E[ y_t -
\beta_{m^1}(\omega_{t}) m^1_{t} - \beta_{m^2}(\omega_{t}) m^2_{t} \ | \
l^s_{t}, l^u_{t}, k_{t} = k, r^{1,2}_{t}] \right\vert_{k=k_{t}}.
\end{align*}
The additive productivity is identified in closed form by 
\begin{equation*}
\beta_{0}(\omega_{t}) = E[ y_t - \beta_{l^s}(\omega_{t}) l^s_{t} -
\beta_{l^u}(\omega_{t}) l^u_{t} - \beta_{k}(\omega_{t}) k_{t} -
\beta_{m^1}(\omega_{t}) m^1_{t} - \beta_{m^2}(\omega_{t}) m^2_{t} \ | \
l_{t}, k_{t}, r^{1,2}_{t}].
\end{equation*}


\subsection{More Flexible Input Variables}

\label{sec:more_flexible_input_variables}

The baseline model includes only two flexible inputs, $m^1_{t}$ and $m^2_{t}$%
. In applications, however, researchers often use more types of flexible
inputs, such as materials, electricity, and fuels \citep[e.g.,][]{LePe03}.
To accommodate such applications with three flexible inputs, for example, we
can augment the gross-output production function in the logarithm (\ref%
{eq:production_function}) as 
\begin{equation}  \label{eq:aug:aug:production_function}
y_{t} = \Psi\left( l_{t},k_{t},m^1_{t},m^2_{t},m^3_{t},\omega_{t} \right) +
\eta_{t} \qquad \E[\eta_t]=0,
\end{equation}
where $m^3_{t}$ is the logarithm of the third flexible input and all the
other variables are the same as in the baseline model. Accordingly, the
Cobb-Douglas form (\ref{eq:short_hand}) is augmented as 
\begin{equation}  \label{eq:aug:aug:short_hand}
\Psi\left(l_t,k,m^1_t,m^2_t,m^3_t,\omega_t\right) = \beta_{l}(\omega_t) l_t
+ \beta_{k}(\omega_t) k_t + \beta_{m^1}(\omega_t) m^1_t +
\beta_{m^2}(\omega_t) m^2_t + \beta_{m^3}(\omega_t) m^3_t +
\beta_{0}(\omega_t)
\end{equation}
with heterogeneous coefficients $(\beta_{l}(\omega_t),\beta_{k}(\omega_t),%
\beta_{m^1}(\omega_t),\beta_{m^2}(\omega_t),\beta_{m^3}(\omega_t),\beta_{0}(%
\omega_t))$. With the following modifications of Assumptions \ref%
{a:flexible_input}, \ref{a:finite}, \ref{a:measurable}, \ref{a:technology},
and \ref{a:independence} adapted to the current augmented model, we can
construct the identification results in similar lines of argument to those
we had for the baseline model.

\newtheorem*{assumption_aug_aug_flexible_input}{Assumption
\ref{a:flexible_input}$''$} 
\begin{assumption_aug_aug_flexible_input}[Flexible Input Choice]
A firm at time $t$ with the state variables $(l_{t},k_{t},\omega_{t})$ chooses the flexible input vector $(m^1_{t},m^2_{t},m^3_{t})$ by the optimization problem
\begin{align*}
\max_{(m^1,m^2,m^3) \in \mathbb{R}_+^3}
&p_{t}^y \exp\left(\Psi\left( l_{t},k_{t},m^1,m^2,m^3,\omega_{t} \right)\right)
E[\exp\left(\eta_{t}\right)]
\\
&-p_{t}^{m^1} \exp\left( m^1 \right)
-p_{t}^{m^2} \exp\left( m^2 \right)
-p_{t}^{m^3} \exp\left( m^3 \right),
\end{align*}
 where $p_{t}^{m^1} > 0$, $p_{t}^{m^2} > 0$, and $p_{t}^{m^3} > 0$ almost surely.
\end{assumption_aug_aug_flexible_input}

\newtheorem*{assumption_aug_aug_finite}{Assumption \ref{a:finite}$''$} 
\begin{assumption_aug_aug_finite}[Finite Solution]
$\beta_{m^1}(\omega_{t})+\beta_{m^2}(\omega_{t})+\beta_{m^3}(\omega_{t}) < 1$ almost surely.
\end{assumption_aug_aug_finite}

\newtheorem*{assumption_aug_aug_measurable}{Assumption
\ref{a:measurable}$''$} 
\begin{assumption_aug_aug_measurable}[Measurability]
$\beta_{l}(\cdot)$, $\beta_{k}(\cdot)$, $\beta_{m^1}(\cdot)$, $\beta_{m^2}(\cdot)$, $\beta_{m^3}(\cdot)$, and $\beta_{0}(\cdot)$ are measurable functions.
\end{assumption_aug_aug_measurable}

\newtheorem*{assumption_aug_aug_technology}{Assumption
\ref{a:technology}$''$} 
\begin{assumption_aug_aug_technology}[Non-Collinear Heterogeneity]
One of the functions, $\omega \mapsto \beta_{m^1}(\omega) / \beta_{m^2}(\omega)$, $\omega \mapsto \beta_{m^2}(\omega) / \beta_{m^3}(\omega)$, $\omega \mapsto \beta_{m^3}(\omega) / \beta_{m^1}(\omega)$, or $\omega \mapsto \left( \beta_{m^1}(\omega) / \beta_{m^2}(\omega), \beta_{m^2}(\omega) / \beta_{m^3}(\omega) \right)$, is measurable and invertible with measurable inverse.
\end{assumption_aug_aug_technology}

\newtheorem*{assumption_aug_aug_independence}{Assumption
\ref{a:independence}$''$} 
\begin{assumption_aug_aug_independence}[Independence]
$(l_{t},k_{t},\omega_{t},p_{t}^{m^1},p_{t}^{m^2},p_{t}^{m^3}) \ind \eta_{t}$.
\end{assumption_aug_aug_independence}

Assumption \ref{a:flexible_input}$^{\prime \prime }$ formally requires that $%
m_{t}^{1}$, $m_{t}^{2}$, and $m_{t}^{3}$ are the three flexible inputs while
the others are state variables. Assumption \ref{a:finite}$^{\prime \prime }$
requires diminishing returns with respect to the three flexible inputs, but
not necessarily with respect to all the production factors. Assumptions \ref%
{a:measurable}$^{\prime \prime }$ and \ref{a:independence}$^{\prime \prime }$
are straightforward modifications of Assumptions \ref{a:measurable} and \ref%
{a:independence}, respectively, suitable for the current model (\ref%
{eq:aug:aug:production_function})--(\ref{eq:aug:aug:short_hand}). Assumption %
\ref{a:technology}$^{\prime \prime }$ is a less straightforward modification
of Assumption \ref{a:technology}, and merits some discussions concerning its
implication for the latent technology. If the assumption holds for one of
the first three maps, namely $\omega \mapsto \beta _{m^{1}}(\omega )/\beta
_{m^{2}}(\omega )$, $\omega \mapsto \beta _{m^{2}}(\omega )/\beta
_{m^{3}}(\omega )$, or $\omega \mapsto \beta _{m^{3}}(\omega )/\beta
_{m^{1}}(\omega )$, then this assumption still requires the latent
technology $\omega _{t}$ to be one-dimensional, and the interpretation of
Assumption \ref{a:technology}$^{\prime \prime }$ is analogous to that of
Assumption \ref{a:technology}. On the other hand, if the assumption holds
for the last map, namely $\omega \mapsto \left( \beta _{m^{1}}(\omega
)/\beta _{m^{2}}(\omega ),\beta _{m^{2}}(\omega )/\beta _{m^{3}}(\omega
)\right) $, then this assumption requires the latent technology $\omega _{t}$
to be two-dimensional. In this case, similar lines of arguments to those for
the baseline model yield the vector of input cost ratios $\left(
r_{t}^{1,2},r_{t}^{2,3}\right) =\left( \beta _{m^{1}}(\omega )/\beta
_{m^{2}}(\omega ),\beta _{m^{2}}(\omega )/\beta _{m^{3}}(\omega )\right) $
as a control variable for the two-dimensional technology $\omega _{t}$,
where $r_{t}^{2,3}$ is defined by $r_{t}^{2,3}=p_{t}^{m^{2}}\exp
(m_{t}^{2})/p_{t}^{m^{3}}\exp (m_{t}^{3})$ analogously to (\ref{eq:ratio}),
and the identifying formulas will thus entail controlling for these two
ratios. We state as a theorem the extended identification result based on
these modified assumptions.

\begin{theorem}[Identification]
Suppose that Assumptions \ref{a:flexible_input}$^{\prime \prime }$, \ref%
{a:finite}$^{\prime \prime }$, \ref{a:measurable}$^{\prime \prime }$, \ref%
{a:technology}$^{\prime \prime }$, and \ref{a:independence}$^{\prime \prime }
$ are satisfied for the model (\ref{eq:aug:aug:production_function})--(\ref%
{eq:aug:aug:short_hand}). The parameter vector $(\beta_{l}(\omega_{t}),%
\beta_{k}(\omega_{t}),\beta_{m^1}(\omega_{t}),\beta_{m^2}(\omega_{t}),%
\beta_{m^3}(\omega_{t}),\beta_{0}(\omega_{t}))$ with the latent productivity 
$\omega_t$ is identified if there exists $(l^s_t,l^u_t,k_t)$ such that $%
(l^s_t,l^u_t,k_t,\omega_t)$ is at a locality of identification (Definition %
\ref{def:locality_of_identification}).
\end{theorem}

The coefficients of flexible inputs, $m^1$ and $m^2$, are identified in
closed form by 
\begin{align*}
\beta_{m^1}(\omega_{t}) &= s^1_{t} = \frac{ p_{t}^{m^1} \exp(m^1) }{ \E[
p_{t}^y \exp(y_{t}) | l_{t}, k_{t}, m^1_{t}, m^2_{t}, m^3_{t}, r^{1,2}_{t},
r^{2,3}_{t} ] }, \\
\beta_{m^2}(\omega_{t}) &= s^2_{t} = \frac{ p_{t}^{m^2} \exp(m^2) }{ \E[
p_{t}^y \exp(y_{t}) | l_{t}, k_{t}, m^1_{t}, m^2_{t}, m^3_{t}, r^{1,2}_{t},
r^{2,3}_{t} ] }, \qquad\text{and} \\
\beta_{m^3}(\omega_{t}) &= s^3_{t} = \frac{ p_{t}^{m^3} \exp(m^3) }{ \E[
p_{t}^y \exp(y_{t}) | l_{t}, k_{t}, m^1_{t}, m^2_{t}, m^3_{t}, r^{1,2}_{t},
r^{2,3}_{t} ] }
\end{align*}
The coefficients of labor $l$ and capital $k$ are identified in turn in
closed form by 
\begin{align*}
\beta_{l}(\omega_{t}) &= \left. \frac{\partial}{\partial l} E[ y_t -
\beta_{m^1}(\omega_{t}) m^1_{t} - \beta_{m^2}(\omega_{t}) m^2_{t} -
\beta_{m^3}(\omega_{t}) m^3_{t} \ | \ l_{t} = l, k_{t}, r^{1,2}_{t},
r^{2,3}_{t}] \right\vert_{l=l_{t}} \qquad\text{and} \\
\beta_{k}(\omega_{t}) &= \left. \frac{\partial}{\partial k} E[ y_t -
\beta_{m^1}(\omega_{t}) m^1_{t} - \beta_{m^2}(\omega_{t}) m^2_{t} -
\beta_{m^3}(\omega_{t}) m^3_{t} \ | \ l_{t}, k_{t} = k, r^{1,2}_{t},
r^{2,3}_{t}] \right\vert_{k=k_{t}}.
\end{align*}
The additive productivity is identified in closed form by 
\begin{equation*}
\beta_{0}(\omega_{t}) = E[ y_t - \beta_{l}(\omega_{t}) l_{t} -
\beta_{k}(\omega_{t}) k_{t} - \beta_{m^1}(\omega_{t}) m^1_{t} -
\beta_{m^2}(\omega_{t}) m^2_{t} - \beta_{m^3}(\omega_{t}) m^3_{t} \ | \
l_{t}, k_{t}, r^{1,2}_{t}, r^{2,3}_{t}].
\end{equation*}


\subsection{Single $m_{t}$ Variable}

\label{sec:single_m_variable}

The baseline model includes two flexible inputs, $m^1_{t}$ and $m^2_{t}$.
Researchers sometimes include only one type of inputs, such as materials,
other than labor and capital. We may accommodate such a reduced model at the
cost of an alternative timing assumption for labor input, namely concurrent
choice of labor input. Write a parsimonious version of the gross-output
production function in the logarithm (\ref{eq:production_function}) as 
\begin{equation}  \label{eq:red:production_function}
y_{t} = \Psi\left( l_{t},k_{t},m_{t},\omega_{t} \right) + \eta_{t} \qquad %
\E[\eta_t]=0,
\end{equation}
Accordingly, the Cobb-Douglas form (\ref{eq:short_hand}) is reduced as 
\begin{equation}  \label{eq:red:short_hand}
\Psi\left(l_t,k,m_t,\omega_t\right) = \beta_{l}(\omega_t) l_t +
\beta_{k}(\omega_t) k_t + \beta_{m}(\omega_t) m_t + \beta_{0}(\omega_t)
\end{equation}
with heterogeneous coefficients $(\beta_{l}(\omega_t),\beta_{k}(\omega_t),%
\beta_{m}(\omega_t),\beta_{0}(\omega_t))$. With the following modifications
of Assumptions \ref{a:flexible_input}, \ref{a:finite}, \ref{a:measurable}, %
\ref{a:technology}, and \ref{a:independence} adapted to the current
augmented model, we can construct the identification results in similar
lines of argument to those we had for the baseline model.

\newtheorem*{assumption_red_flexible_input}{Assumption
\ref{a:flexible_input}$'''$} 
\begin{assumption_red_flexible_input}[Flexible Input Choice]
A firm at time $t$ with the state variables $(k_{t},\omega_{t})$ chooses the flexible input vector $(l_{t},m_{t})$ by the optimization problem
\begin{align*}
\max_{(l,m) \in \mathbb{R}_+^2}
&p_{t}^y \exp\left(\Psi\left( l,k_{t},m,\omega_{t} \right)\right)
E[\exp\left(\eta_{t}\right)]
\\
&-p_{t}^{l} \exp\left( l \right)
-p_{t}^{m} \exp\left( m \right),
\end{align*}
 where $p_{t}^{l} > 0$ and $p_{t}^{m} > 0$ almost surely.
\end{assumption_red_flexible_input}

\newtheorem*{assumption_red_finite}{Assumption \ref{a:finite}$'''$} 
\begin{assumption_red_finite}[Finite Solution]
$\beta_{l}(\omega_{t})+\beta_{m}(\omega_{t}) < 1$ almost surely.
\end{assumption_red_finite}

\newtheorem*{assumption_red_measurable}{Assumption \ref{a:measurable}$'''$} 
\begin{assumption_red_measurable}[Measurability]
$\beta_{l}(\cdot)$, $\beta_{k}(\cdot)$, $\beta_{m}(\cdot)$, and $\beta_{0}(\cdot)$ are measurable functions.
\end{assumption_red_measurable}

\newtheorem*{assumption_red_technology}{Assumption \ref{a:technology}$'''$} 
\begin{assumption_red_technology}[Non-Collinear Heterogeneity]
The function $\omega \mapsto \beta_{l}(\omega) / \beta_{m}(\omega)$ is measurable and invertible with measurable inverse.
\end{assumption_red_technology}

\newtheorem*{assumption_red_independence}{Assumption
\ref{a:independence}$'''$} 
\begin{assumption_red_independence}[Independence]
$(l_{t},k_{t},\omega_{t},p_{t}^{l},p_{t}^{m}) \ind \eta_{t}$.
\end{assumption_red_independence}

Assumption \ref{a:flexible_input}$^{\prime \prime \prime }$ asserts that,
unlike the baseline model, the labor input $l_{t}$ is treated as a flexible
input along with the materials $m_{t}$, as opposed to a predetermined
quantity. This will \textit{not} incur the functional independence problem %
\citep{AcCaFr15} because the output elasticity with respect to labor in this
case is unambiguously identified via the first-order condition just like the
output elasticity with respect to materials. Accordingly, the locality of
identification defined for the current model below does \textit{not} require
data variations in $l_{t}$ given the state variables fixed. The
identification argument in the current reduced model relies on the
non-collinear heterogeneity in the ratio of the elasticity with respect to
labor to the elasticity with respect to materials. With the following
modification to the definition of the locality of identification adapted to
the current augmented model, we state as a theorem the identification result.

\newtheorem*{definition_red_locality_of_identification}{Definition
\ref{def:locality_of_identification}$'''$} 
\begin{definition_red_locality_of_identification}[Locality of Identification]
The point $(k_{t},\omega_{t}) = (k^\ast,\omega^\ast)$ is called the locality of identification if there are real numbers $\underline{f}, \overline{f} \in (0,\infty)$ such that $\underline{f} < f_{k_{t},\omega_{t}}(k,\omega) < \overline{f}$ for all $(k,\omega)$ in a neighborhood of $(k^\ast,\omega^\ast)$.
\end{definition_red_locality_of_identification}

\begin{theorem}[Identification]
Suppose that Assumptions \ref{a:flexible_input}$^{\prime \prime \prime }$, %
\ref{a:finite}$^{\prime \prime \prime }$, \ref{a:measurable}$^{\prime \prime
\prime }$, \ref{a:technology}$^{\prime \prime \prime }$, and \ref%
{a:independence}$^{\prime \prime \prime }$ are satisfied for the model (\ref%
{eq:red:production_function})--(\ref{eq:red:short_hand}). The parameter
vector $(\beta_{l}(\omega_{t}),\beta_{k}(\omega_{t}),\beta_{m}(\omega_{t}),%
\beta_{0}(\omega_{t}))$ with the latent productivity $\omega_t$ is
identified if there exists $k_t$ such that $(k_t,\omega_t)$ is at a locality
of identification (Definition \ref{def:locality_of_identification}$^{\prime
\prime \prime }$).
\end{theorem}

With the flexible input cost ratio modified as 
\begin{equation*}
r^{l,m}_{t} = \frac{p_{t}^{l} \exp(l_{t})}{p_{t}^{m} \exp(m_{t})}.
\end{equation*}
the coefficients of flexible inputs, $l$ and $m$, are identified in closed
form by 
\begin{align*}
\beta_{l}(\omega_{t}) &= s^l_{t} = \frac{ p_{t}^{l} \exp(l) }{ \E[ p_{t}^y
\exp(y_{t}) | l_{t}, k_{t}, m_{t}, r^{l,m}_{t} ] } \qquad\text{and} \\
\beta_{m}(\omega_{t}) &= s^m_{t} = \frac{ p_{t}^{m} \exp(m) }{ \E[ p_{t}^y
\exp(y_{t}) | l_{t}, k_{t}, m_{t}, r^{l,m}_{t} ] }
\end{align*}
The coefficient capital $k$ is identified in turn in closed form by 
\begin{align*}
\beta_{k}(\omega_{t}) &= \left. \frac{\partial}{\partial k} E[ y_t -
\beta_{l}(\omega_{t}) l_{t} - \beta_{m}(\omega_{t}) m_{t} \ | \ k_{t} = k,
r^{l,m}_{t}] \right\vert_{k=k_{t}}.
\end{align*}
The additive productivity is identified in closed form by 
\begin{equation*}
\beta_{0}(\omega_{t}) = E[ y_t - \beta_{l}(\omega_{t}) l_{t} -
\beta_{k}(\omega_{t}) k_{t} - \beta_{m}(\omega_{t}) m_{t} \ | \ k_{t},
r^{l,m}_{t}].
\end{equation*}


\section{Summary and Discussions}\label{sec:summary}

\label{sec:summary} In this paper, we develop the identification of
heterogeneous elasticities in the Cobb-Douglas production function. The
identification is constructively achieved with closed-form formulas for the
output elasticity with respect to each input, as well as the additive
productivity, for each firm. The flexible input cost ratio plays the role of
a control function under the assumption of non-collinear heterogeneity
between elasticities with respect to two flexible inputs. The \textit{ex ante%
} flexible input cost share is shown to be useful to identify the
elasticities with respect to flexible inputs for each firm. The elasticities
with respect to labor and capital can be identified for each firm under the
timing assumption admitting the functional independence. Extended
identification results are provided for three alternative models that are
frequently used in empirical analysis.

In light of the fact that conventional identification strategies for production functions use panel data, it is unusual for our identification strategy not to rely on panel data.
Note that the existing papers use panel data to form orthogonality restrictions to estimate input coefficients.
Our explicit identification of the flexible input coefficients via the first-order conditions does not involve any panel structure.
This feature entails a couple of limitations which are the costs that we pay for not relying on panel data and for our ability to identify heterogeneous elasticities.
While these limitations are shared by other recent papers that also use the first-order conditions for identification, we discuss them below and propose a scope of future research.

The first limitation of our identification method is the requirement for input and output prices, which are not always available in empirical data.
In certain types of production analysis, however, the prices are normalized to one and thus are assumed to be known.
For example, many important papers, including \cite{LePe03}, in estimation of production functions use a data set that is based on the census for plants collected by Chile's Instituto Nacional de Estadistica \cite[cf.][]{Lu91}.
For this data set, the output ($\exp(y_{t})$) is the gross revenue deflated to real Chilean pesos in a baseline year.
Flexible inputs include measures of materials ($\exp(m^1_{t})$), electricity ($\exp(m^2_{t})$), and fuels ($\exp(m^3_{t})$) all measured in terms of pecuniary values deflated to real Chilean pesos in a baseline year.
Since they are measured in terms of values as opposed to quantities, the researchers effectively set $p_t^y=p_t^{m^1}=p_t^{m^2}=p_t^{m^3}=1$ in their analysis.
For this data set, therefore, our first limitation is not binding.

The second limitation of our identification approach is the assumption of price-taking firms in both the output market and the flexible input markets.
This assumption rules out market power that is relevant to answering some important policy questions in industry studies and international trade \citep[e.g.,][]{DeWa12,DeGoKhPa16}.
We leave identification strategies under market power for future research.

\newpage


\begin{thebibliography}{9}
\bibitem[Ackerberg, Caves, and Frazer(2015)]{AcCaFr15} Ackerberg, D.A.,
Caves, K., and Frazer, G. (2015) Identification properties of recent
production function estimators. \textit{Econometrica} \textbf{83} 2411--2451.

\bibitem[Ackerberg, Benkard, Berry, and Pakes(2007)]{AcBeBePa07} Ackerberg,
D., Benkard, L., Berry, S. and Pakes, A. (2007) Econometric tools for
analyzing market outcomes. in J. Heckman \& E. Leamer, eds, \textit{Handbook
of Econometrics} \textbf{6A} 4171--4276.

\bibitem[Ackerberg and Hahn(2015)]{AcHa15}
Ackerberg, D. and Hahn, J. (2015) Some non-parametric identication results using timing and information set assumptions. Working Paper.

\bibitem[Balat, Brambilla, and Sasaki(2016)]{BaBrSa16} Balat, J. Brambilla,
I., and Sasaki, Y. (2016) Heterogeneous firms: skilled-labor productivity
and the destination of exports. Working Paper.

\bibitem[Chen(2007)]{Ch07} Chen, X. (2007) Large sample sieve estimation of
semi-nonparametric models. in J. Heckman \& E. Leamer, eds, \textit{Handbook
of Econometrics} \textbf{6A} 5549--5632.

\bibitem[De Loecker and Warzynski(2012)]{DeWa12}
De Loecker, J. and Warzynski, F. (2012)
Markups and firm-level export status.
\textit{American Economic Review} \textbf{102} 2437--2471.

\bibitem[De Loecker, Goldberg, Khandelwal, and Pavcnik(2016)]{DeGoKhPa16}
De Loecker, J., Goldberg, P.K., Khandelwal, A.K., and Pavcnik, N. (2016)
Prices, markups, and trade reform.
\textit{Econometrica} \textbf{84} 445--510.

\bibitem[Doraszelski and Jaumandreu(2013)]{DoJa13} Doraszelski, U. and
Jaumandreu, J. (2013) R\&D and productivity: estimating endogenous
productivity. \textit{Review of Economic Studies} \textbf{80} 1338--1383.

\bibitem[Doraszelski and Jaumandreu(2015)]{DoJa15} Doraszelski, U. and
Jaumandreu, J. (2015) Measuring the bias of technological change. Working
Paper.

\bibitem[Durrett(2010)]{Durrett2010} Durrett, R. (2010) \textit{Probability:
Theory and Examples.} Cambridge University Press.

\bibitem[Gandhi, Navarro, and Rivers(2017)]{GaNaRi17} Gandhi, A., Navarro,
S., and Rivers, D. (2017) How heterogeneous is productivity? A comparison of
gross output and value added. \textit{Journal of Political Economy}
(forthcoming).

\bibitem[Grieco, Li, and Zhang(2016)]{GrLiZh16} Grieco, P., Li, S., and
Zhang, H. (2016) Production function estimation with unobserved input price
dispersion. \textit{International Economic Review} \textbf{57} 665--690.

\bibitem[Griliches and Mairesse(1998)]{GrMa98} Griliches, Z. and Mairesse,
J. (1998) Production functions: the search for identification. in \textit{%
Econometrics and Economic Theory in the Twentieth Century: The Ragnar Frisch
Centennial Symposium.} Cambridge University Press, 169--203.

\bibitem[Hu, Huang, and Sasaki(2017)]{HuHuSa17} Hu, Y., Huang, G., and
Sasaki, Y. (2017) Estimating production functions with robustness against
errors in the proxy variables. Working Paper.

\bibitem[Kasahara, Schrimpf, and Suzuki(2015)]{KaScSu15} Kasahara, H.,
Schrimpf, P., and Suzuki, M. (2015) Identification and estimation of
production function with unobserved heterogeneity. Working Paper.

\bibitem[Levinsohn and Petrin(2003)]{LePe03} Levinsohn, J. and Petrin, A.
(2003) Estimating production functions using inputs to control for
unobservables. \textit{Rev. Econ. Stud.} \textbf{70} 317--341.

\bibitem[Lui(1991)]{Lu91}
Lui, L. (1991) Entry-exit and productivity changes: an empirical analysis of efficiency frontiers. Ph.D. Thesis, University of Michigan.

\bibitem[Marschak and Andrews(1944)]{MaAn44} Marschak, J. and Andrews, W.
(1944) Random simultaneous equations and the theory of production. \textit{%
Econometrica} \textbf{12} 143--205.

\bibitem[Neyman and Scott(1948)]{NeSc48} Neyman, J. and Scott, E.L. (1948)
Consistent estimation from partially consistent observations. \textit{%
Econometrica} \textbf{16} 1--32.

\bibitem[Olley and Pakes(1996)]{OlPa96} Olley, G.S. and Pakes, A. (1996) The
dynamics of productivity in the telecommunications equipment industry. 
\textit{Econometrica} \textbf{64} 1263--1297.

\bibitem[Pearl(2000)]{Pearl2000} Pearl, J. (2000) \textit{Causality: Models,
Reasoning and Inference.} Cambridge University Press.

\bibitem[Solow(1957)]{So57} Solow, R.M. (1957) Technical change and the
aggregate production function. \textit{Review of Economics and Statistics} 
\textbf{39} 312--320.

\bibitem[van Biesebroeck(2003)]{Va03} van Biesebroeck, J. (2003)
Productivity dynamics with technology choice: an application to automobile
assembly. \textbf{70} 167--198.

\bibitem[Wooldridge(2009)]{Wo09} Wooldridge, J.M. (2009) On estimating
firm-level production functions using proxy variables to control for
unobservables. \textit{Economics Letters} \textbf{104} 112--114.
\end{thebibliography}
\end{document}